\documentclass[10pt,twocolumn,twoside]{IEEEtran}

\IEEEoverridecommandlockouts                              	

\usepackage{amssymb,amsmath,amsfonts,amsthm}
\usepackage{algorithm,algorithmic}
\usepackage{cite}
\usepackage{float}
\usepackage{epsfig}
\usepackage{graphicx}
\usepackage{graphics}
\usepackage{subfigure}
\usepackage{url}
\usepackage{xspace}
\usepackage[usenames,dvipsnames]{color}
\usepackage{soul}
\usepackage{mathtools}

\newcommand{\TC}[1]{\textcolor{cyan}{{#1}}}

\floatstyle{ruled}
\newfloat{model}{H}{mod}
\floatname{model}{\footnotesize Model}
\newfloat{notatio}{H}{not}
\floatname{notatio}{\footnotesize Notation}

\newenvironment{varalgorithm}[1]
  {\algorithm\renewcommand{\thealgorithm}{#1}}
  {\endalgorithm}

\newenvironment{list4}{
	\begin{list}{$\bullet$}{%
			\setlength{\itemsep}{0.05cm}
			\setlength{\labelsep}{0.2cm}
			\setlength{\labelwidth}{0.3cm}
			\setlength{\parsep}{0in} 
			\setlength{\parskip}{0in}
			\setlength{\topsep}{0in} 
			\setlength{\partopsep}{0in}
			\setlength{\leftmargin}{0.16in}}}
	{\end{list}}

\newenvironment{list4a}{
	\begin{list}{$\bullet$}{%
			\setlength{\itemsep}{0.05cm}
			\setlength{\labelsep}{0.2cm}
			\setlength{\labelwidth}{0.3cm}
			\setlength{\parsep}{0in} 
			\setlength{\parskip}{0in}
			\setlength{\topsep}{0in} 
			\setlength{\partopsep}{0in}
			\setlength{\leftmargin}{0.16in}}}
	{\end{list}}

\newenvironment{list5}{
	\begin{list}{$\bullet$}{%
			\setlength{\itemsep}{0.05cm}
			\setlength{\labelsep}{0.2cm}
			\setlength{\labelwidth}{0.3cm}
			\setlength{\parsep}{0in} 
			\setlength{\parskip}{0in}
			\setlength{\topsep}{0in} 
			\setlength{\partopsep}{0in}
			\setlength{\leftmargin}{0.18in}}}
	{\end{list}}

\usepackage{url,changebar,bm,xspace,dsfont}
\let\mathbb=\mathds 

\newtheorem{theorem}{Theorem}

\newtheorem{assum}{Assumption}

\newtheorem{remark}{Remark}

\ifCLASSINFOpdf
 \else
\fi

\usepackage{amsmath,amsfonts,amssymb,amscd}
\usepackage{capt-of}
\usepackage{color}
\usepackage{cite}
\usepackage{verbatim}
\usepackage{hyperref}

\hypersetup{
    colorlinks = true,
    linkcolor = [rgb]{0,0,1},
    anchorcolor = [rgb]{0,0,1},
    citecolor = [rgb]{0.9,0.5,0},
    filecolor = [rgb]{0,0,1},
    urlcolor = [rgb]{0.9,0.5,0},
    bookmarksopen = true,
    bookmarksnumbered = true,
    breaklinks = true,
    linktocpage,
    colorlinks = true,
    linkcolor = [rgb]{0.2,0.6,0.2},
    urlcolor  = [rgb]{0,0,1},
    citecolor = [rgb]{0.9,0.5,0},
    anchorcolor = [rgb]{0.2,0.6,0.2},
}

\newtheorem{lemma}{\bfseries Lemma}

\begin{document}

\title{\LARGE \bf Distributed Optimization for Quadratic Cost Functions over Large-Scale Networks with Quantized Communication and Finite-Time Convergence}


\author{Apostolos~I.~Rikos, Andreas Grammenos, Evangelia Kalyvianaki, \\ Christoforos N. Hadjicostis, Themistoklis Charalambous, and Karl~H.~Johansson
\thanks{Apostolos~I.~Rikos and K.~H.~Johansson are with the Division of Decision and Control Systems, KTH Royal Institute of Technology, SE-100 44 Stockholm, Sweden. They are also affiliated with Digital Futures.  E-mails: {\tt \{rikos,kallej\}@kth.se}.}
\thanks{A. Grammenos and E. Kalyvianaki are with the Department of Computer Science and Technology, University of Cambridge, Cambridge, and the Alan Turing Institute, London, UK.  E-mails:{\tt~\{ag926,ek264\}@cl.cam.ac.uk}.}
\thanks{C. N. Hadjicostis and T.~Charalambous are with the Department of Electrical and Computer Engineering, School of Engineering, University of Cyprus, 1678 Nicosia, Cyprus.  E-mails:{\tt~\{chadjic,charalambous.themistoklis\}@ucy.ac.cy}.}
\thanks{Parts of the results for optimal task scheduling appear in \cite{2021:Rikos_Grammenos_Kalyvianaki_Hadj_Themis_Johan_CPU}. 
The present version of the paper includes (i) an improved version of the algorithm presented in \cite{2021:Rikos_Grammenos_Kalyvianaki_Hadj_Themis_Johan_CPU} which imposes less operational requirements, (ii) full proofs regarding the completion of the proposed algorithm (not provided in \cite{2021:Rikos_Grammenos_Kalyvianaki_Hadj_Themis_Johan_CPU}), (iii) a fully asynchronous algorithm which operates by performing max-consensus in an asynchronous fashion, (iv) an application over federated learning systems, and (v) an extended analysis regarding the operation of our algorithms in large-scale networks such as data centers.}
\thanks{Part of this work was supported by the Knut and Alice Wallenberg Foundation, the Swedish Research Council, and the Swedish Foundation for Strategic Research. 
The work of T. Charalambous was partly supported by the European Research Council (ERC) Consolidator Grant MINERVA (Grant agreement No. 101044629).
}
}

\maketitle
\thispagestyle{empty}
\pagestyle{empty}

%
%
\begin{abstract} 
We propose two distributed iterative algorithms that can be used to solve, in finite time, the distributed optimization problem over quadratic local cost functions in large-scale networks. 
The first algorithm exhibits synchronous operation whereas the second one exhibits asynchronous operation. 
Both algorithms share salient features. 
Specifically, the algorithms operate exclusively with quantized values, which means that the information stored, processed and exchanged between neighboring nodes is subject to deterministic uniform quantization. 
The algorithms rely on event-driven updates in order to reduce energy consumption, communication bandwidth, network congestion, and/or processor usage. 
Finally, once the algorithms converge, nodes distributively terminate their operation. 
We prove that our algorithms converge in a finite number of iterations to the exact optimal solution depending on the quantization level, and we present applications of our algorithms to (i) optimal task scheduling for data centers, and (ii) global model aggregation for distributed federated learning. 
We provide simulations of these applications to illustrate the operation, performance, and advantages of the proposed algorithms. 
Additionally, it is shown that our proposed algorithms compare favorably to algorithms in the current literature. 
Quantized communication and asynchronous updates increase the required time to completion, but finite-time operation is maintained. 
\end{abstract}


%
%
%
%
\section{Introduction}\label{intro}



In several applications (such as, cloud computing, power systems, and traffic networks), the traditional way of optimizing performance is to collect data from all users (demands, requests, etc.) at a central agent/processor (usually the service provider).
This processor performs the (usually heavy) computations for the optimization centrally, and distributes the corresponding part of the solution to each of the users. As a consequence, centralized optimization algorithms (see, e.g., \cite{mao_learning_2019, 2009:Isard_Goldberg, 2016:Tumanov_Ganger, 2016:Grandl_Kulkarni}), anticipate the existence of a central processor, and require transfer of information to it in order to perform the necessary computations. 
However, centralized algorithms are subject to limitations as they (i) increase the overall network traffic (with a bottleneck at the central agent who collects the information from all users), (ii) lack scalability, (iii) require high computational capabilities for the centralized server, and (iv) can be subject to a single point of failure \cite{2020:Themis_Kalyvianaki}. 

In recent years, rapid developments of telecommunication systems, has given rise to new networked systems applications such as, distributed localization, distributed estimation, and distributed target tracking.  
In these applications several agents interact with each other via wired/wireless channels for completing a task. 
As a result, there has been growing interest for control and coordination of networked systems. 
One such problem is that of distributed optimization \cite{Tao:2019Survey, 2009:Nedic_Optim}, in which agents are required to operate in a cooperative fashion in order to minimize/maximize the desirable cost/utility function. 
Distributed operation alleviates the limitations of centralized optimization algorithms and facilitates new disruptive applications. 
%
In distributed optimization algorithms, each node has access to a local objective function which is part of the global objective function. 
Then, it aims to optimize the global objective function by optimizing its own local objective function while coordinating with other nodes in the network in a peer-to-peer fashion. 
Distributed operation offers various advantages such as (i) scalability, (ii) resiliency to single points of failure, (iii) performance improvements, and (iv) more efficient (and distributed) usage of network resources.  

Most distributed optimization algorithms in the literature require nodes to exchange real\TC{-}valued messages to guarantee convergence to the desired solution \cite{2015_Alej_Hadj, Tao:2019Survey, 2013:CDC1Themis, 2009:Nedic_Optim, 2018:Xie, 2018:Khan_addopt, 2021:Nedic_PushPull}. 
Due to their real-valued operation, they converge to the optimal solution in an asymptotic fashion. 
Hence, such algorithms are not suitable in practice, since they operate under some unrealistic assumptions. More specifically, real-valued messages assume channels of infinite capacity and asymptotic convergence means that the optimal is obtained as time approaches infinity. 
If they operate over a limited time-frame, then they terminate in the proximity of the optimal solution \cite{2020:Themis_Kalyvianaki, 2021:Liu_Yu} 
This is a direct consequence of the asymptotic nature of their operation, which does not provide finite-time convergence guarantees.  
Some recent works have focused on algorithms with finite-time convergence which allow calculation of the exact solution without any error \cite{2022:Jiang_Charalambous}. 
However, during the operation of these algorithms the values exchanged between nodes should be real numbers, which increases the communication overhead bottleneck (as channels of infinite capacity are required). 
%
%
In order to achieve more efficient usage of network resources (such as bandwidth), communication-efficient distributed optimization (where the exchanged values are quantized) has received significant attention recently in the control and machine learning communities \cite{2020:Jadbabaie_Federated, 2019:Koloskova_Jaggi, 2018:Lalitha_Koushanfar}. 
However, most current approaches are mainly quantizing the values of a real valued distributed optimization algorithm and thus, they maintain the algorithm's asymptotic convergence nature. 
This means that even though they operate with quantized values, they converge to the optimal solution in an asymptotic fashion. 
To the best of our knowledge, there is no algorithm in the literature that utilizes quantized values for efficient communication and converges in finite time to the proximity of the optimal solution or the exact one, depending on the quantization level (i.e., the distance of the calculated and the exact solution is always less than the quantization level).
The work in this paper combines both characteristics, and aims at paving the way for the use of finite-time algorithms that operate solely with quantized values to address distributed optimization problems in directed networks.

\vspace{.2cm} 

\textbf{Main Contributions.} 
In this paper, we focus on distributed optimization over quadratic local cost functions. 
We take into consideration that the exchanged messages consist of quantized values, and propose two distributed algorithms -- one synchronous and one asynchronous -- that solve the optimization problem in a finite number of steps. 
The algorithms terminate their operation once they converge to the optimal solution. 
The main contributions of the paper are the following.
\begin{list5}
\item We present a novel synchronous distributed algorithm that solves the optimization problem in a finite number of time steps using quantized values (Algorithm~\ref{algorithm1}). 
A distributed stopping mechanism is deployed in order to terminate the algorithm in a finite number of time steps.
\item We provide an upper bound on the number of time steps needed for convergence based on properties of primitive matrices. 
The convergence time relies on the network connectivity as determined by the diameter of the network, rather than the number of network nodes (Theorem~\ref{thm:main}). 
\item We present a fully asynchronous distributed algorithm that solves the optimization problem in a finite number of time steps using quantized values (Algorithm~\ref{algorithm2}). 
Specifically, we present a distributed stopping mechanism which relies on asynchronous updates. 
Then, we discuss a modified asynchronous algorithm for the case where each node requires a random number of time steps to process the received information. 
\item We provide an upper bound on the number of time steps needed for convergence of our asynchronous algorithm. 
The convergence time depends on connectivity and the number of time steps required by each node to process the received information (Theorem~\ref{thm:main_2}). 
\item Even though the proposed algorithms could be used in many applications, here, we discuss them within the context of (i) resource management in cloud infrastructures \cite{2020:Themis_Kalyvianaki}, and (ii) global model aggregation in distributed federated learning systems \cite{2020:Jadbabaie_Federated, 2020:Virginia, 2019:Roselander, 2018:Lalitha_Koushanfar, 2020:Miao, 2019:Hu_Wang} (Section~\ref{sec:applications}).
One such application is task scheduling which is a problem that optimally allocates tasks to server machines. 
The goal is to balance the central processing unit (CPU) utilization across data center servers by carefully deciding how to allocate tasks to CPU resources in a distributed fashion \cite{2020:Themis_Kalyvianaki, 2021:Rikos_Grammenos_Kalyvianaki_Hadj_Themis_Johan_CPU, 2020:Doostmohammadian_Charalambous}. 
The other application is global model aggregation, where the local calculated models are aggregated over a federated learning system. 
The goal is to calculate the global model parameters by aggregating the local model parameters of each node (obtained by local training on each node) in a distributed fashion \cite{2020:Miao, 2020:Virginia, 2019:Hu_Wang}. 
These applications are discussed in detail in the paper. 
\end{list5}



The operation of our proposed algorithms relies on similar principles as \cite{2015:Cady, 2010:christoforos, 2021:Rikos_Hadj_Johan, 2013:Giannini}. 
Specifically, our proposed algorithms rely on calculating the average of a given set of variables, and then terminating their operation. 
However, the operation of our algorithms is substantially different than the ones in \cite{2015:Cady, 2010:christoforos, 2021:Rikos_Hadj_Johan}. 
The main differences are that the proposed algorithms (i) operate with quantized values, (ii) exhibit-finite time convergence, and (iii) are able to terminate their operation once convergence has been achieved. 
Furthermore, the proposed algorithms' distributed stopping mechanism shares similarities with \cite{2013:Giannini}. 
However, our proposed distributed stopping mechanism (unlike the one in \cite{2013:Giannini}) utilizes two parallel linear iterations which are adjusted to the quantized nature of our algorithms. 


\textbf{Paper Organization.}
The remainder of the paper is organized as follows. 
In Section~\ref{sec:preliminaries}, we introduce the notation used in the paper and in Section~\ref{sec:probForm}, we provide the problem formulation. 
In Section~\ref{sec:distr_algo}, we present the proposed quantized optimal allocation algorithms and we analyze their convergence. 
In Section~\ref{sec:applications}, we present applications of the algorithms and compare their operation against other algorithms in the literature.
Finally, in Section~\ref{sec:conclusions}, we conclude this paper and discuss possible future directions.

%
%
%
%
\section{Notation and Preliminaries}\label{sec:preliminaries}



\textbf{Notation.}
The sets of real, rational, integer, and natural numbers are denoted by $ \mathbb{R}, \mathbb{Q}, \mathbb{Z}$, and $\mathbb{N}$, respectively. 
Symbol $\mathbb{Z}_{\geq 0}$ ($\mathbb{Z}_{>0}$) denotes the set of nonnegative (positive) integer numbers, while $\mathbb{Z}_{\leq 0}$ ($\mathbb{Z}_{<0}$) denotes the set of nonpositive (negative) integer numbers. 
For any real number $a \in \mathbb{R}$, the floor $\lfloor a \rfloor$ denotes the greatest integer less than or equal to $a$ while the ceiling $\lceil a \rceil$ denotes the least integer greater than or equal to $a$. 
Vectors are denoted by small letters, matrices are denoted by capital letters, and the transpose of a matrix $A$ is denoted by $A^T$. 
For a matrix $A\in \mathbb{R}^{n\times n}$, the entry at row $i$ and column $j$ is denoted by $A_{ij}$.
By $\mathbf{1}$ we denote the all-ones column vector and by $I$ we denote the identity matrix (of appropriate dimensions). 

\textbf{Network Preliminaries.} The communication topology of a network of $n$ ($n \geq 2$) nodes communicating only with their immediate neighbors can be captured by a directed graph (digraph) defined as $\mathcal{G}_d = (\mathcal{V}, \mathcal{E})$. 
In digraph $\mathcal{G}_d$, $\mathcal{V} =  \{v_1, v_2, \dots, v_n\}$ is the set of nodes, whose cardinality is denoted as $| \mathcal{V} |  = n $, and $\mathcal{E} \subseteq \mathcal{V} \times \mathcal{V} - \{ (v_j, v_j) \ | \ v_j \in \mathcal{V} \}$ is the set of edges (self-edges excluded) whose cardinality is denoted as $m = | \mathcal{E} |$. 
A directed edge from node $v_i$ to node $v_j$ is denoted by $m_{ji} \triangleq (v_j, v_i) \in \mathcal{E}$, and captures the fact that node $v_j$ can receive information from node $v_i$ (but not the other way around). 
We assume that the given digraph $\mathcal{G}_d = (\mathcal{V}, \mathcal{E})$ is \textit{strongly connected}, i.e., for each pair of nodes $v_j, v_i \in \mathcal{V}$, $v_j \neq v_i$, there exists a directed \textit{path} from $v_i$ to $v_j$. 
A directed \textit{path} of length $t$ from $v_i$ to $v_j$ exists if we can find a sequence of nodes $v_i \equiv v_{l_0},v_{l_1}, \dots, v_{l_t} \equiv v_j$ such that $(v_{l_{\tau+1}},v_{l_{\tau}}) \in \mathcal{E}$ for $ \tau = 0, 1, \dots , t-1$. 
Furthermore, the diameter $D$ of a digraph is the longest shortest path between any two nodes $v_j, v_i \in \mathcal{V}$ in the network. 
The subset of nodes that can directly transmit information to node $v_j$ is called the set of in-neighbors of $v_j$ and is represented by $\mathcal{N}_j^- = \{ v_i \in \mathcal{V} \; | \; (v_j,v_i)\in \mathcal{E}\}$. 
The cardinality of $\mathcal{N}_j^-$ is called the \textit{in-degree} of $v_j$ and is denoted by $\mathcal{D}_j^-$. 
The subset of nodes that can directly receive information from node $v_j$ is called the set of out-neighbors of $v_j$ and is represented by $\mathcal{N}_j^+ = \{ v_l \in \mathcal{V} \; | \; (v_l,v_j)\in \mathcal{E}\}$. 
The cardinality of $\mathcal{N}_j^+$ is called the \textit{out-degree} of $v_j$ and is denoted by $\mathcal{D}_j^+$.



%
%
%
%
\section{Problem Formulation}\label{sec:probForm}




Let us consider a strongly connected network $\mathcal{G}_d = (\mathcal{V}, \mathcal{E})$. We focus on the scenario where nodes in a network cooperatively minimize a common additive cost function. 
Traditionally, each one of the $n  = | \mathcal{V} |$ nodes is endowed with (and has information only for) a scalar quadratic local cost function $f_i : \mathbb{R} \mapsto \mathbb{R}$. 
Since, in this work, we consider the exchange of integer\footnote{We assume that states are integer-valued which captures a class of quantization effects such as uniform quantization.} values, the local cost function takes rational numbers as inputs. 
Additionally, since the update of each node is also quantized, the outputs are also rational, i.e., $f_i : \mathbb{Q} \mapsto \mathbb{Q}$.   
Nodes aim to cooperatively solve the optimization problem, herein called P1, in finite time and terminate their operation once calculating the optimal solution. P1 is as follows: 
\begin{subequations} 
\begin{align}\label{Global_cost_functions} 
\mathbf{P1:}~\min_{x \in \mathbb{Q}^n}~ & f(x_1, x_2, ..., x_n) \equiv \sum_{i=1}^n f_i(x_i),  \\ 
\text{s.t.}~ & x_i = x_j, \forall v_i, v_j, \in \mathcal{V},  \\ 
       & \text{nodes exchange integer values.} 
\end{align}
\end{subequations}


In this work, we restrict our attention to a quadratic local cost function for every node $v_i$ of the following form cf., \cite{2004:Sensors, Tao:2019Survey}
\begin{equation}\label{local_cost_functions}
    f_i(x_i) = \dfrac{1}{2} \alpha_i (x_i - \rho_i)^2 , 
\end{equation}
where $\alpha_i \in \mathbb{Q}$ and $\rho_i \in \mathbb{Q}$ are given parameters. 
This cost function represents the demand at node $v_i$ with $x_i$ being a global optimization parameter that will determine the optimal solution at each node. 
Note that the choice of quadratic local cost functions allows us 
to calculate a closed-form expression of the optimal solution which can by computed distributively by applying consensus algorithms. 
Specifically, optimization problem \eqref{Global_cost_functions} can be solved in closed form and the optimal solution $x^*$ is given by
\begin{align}\label{eq:closedform}
x^* =  \frac{\sum_{i=1}^n \alpha_i \rho_{i}}{\sum_{i=1}^n \alpha_i}.  
\end{align}
Since $\alpha_i \in \mathbb{Q}$, $\rho_i \in \mathbb{Q}$, then $x^* \in \mathbb{Q}$. 
(Note that if $\alpha_i \in \mathbb{Q}$ and $\rho_i \in \mathbb{Q}$, we can use simple transformations so that we can transform $\alpha_i \in \mathbb{Z}$ and $\rho_i \in \mathbb{Z}$; for simplicity of exposition we adopt this assumption, i.e., $\alpha_i \in \mathbb{Z}$ and $\rho_i \in \mathbb{Z}$. 
Furthermore, we assume that the initial values of the states, $x_i[0]$ are integers, i.e.,  $x_i[0]\in X_0 \subset \mathbb{Z}$ (e.g., $x[0]$ can be the CPU utilization percentage of a server). 
Furthermore, note that our proposed algorithm's calculation of the exact optimal solution depends on the quantization level. 
This means that the distance of the calculated and the exact solution is always less than the quantization level.

\section{Quantized Distributed Solution}
\label{sec:distr_algo}


In this section we propose two distributed quantized information exchange algorithms that solve the problem described in Section~\ref{sec:probForm}. 
The proposed algorithms are detailed as Algorithm~\ref{algorithm1} and Algorithm~\ref{algorithm2}, and they calculate $x^*$ shown in \eqref{eq:closedform} for the case where the updates of every node during the algorithm’s operation are synchronous and asynchronous, respectively.
Both algorithms allow the calculation of $x^*$ after a finite number of time steps. 
In order to solve the problem in a distributed way we make the following assumptions. 

\begin{assum}\label{str_digr}
The communication topology is a strongly connected digraph. 
\end{assum}

\begin{assum}\label{Diam_known}
The diameter of the network $D$ (or an upper bound $D'$) is known to all nodes $v_j \in \mathcal{V}$.
\end{assum}

Assumption~\ref{str_digr} is a necessary condition for each node $v_j$ to be able to calculate the optimal allocation after a finite number of time steps. 
Assumption~\ref{Diam_known} is necessary for the operation of our algorithm and for coordinating the $\min$- and $\max$-consensus algorithms, such that each node $v_j$ is able to determine whether convergence has been achieved and thus the operation can be terminated. 

\subsection{Synchronous Quantized Distributed Solution}\label{alg_code}

We now describe the main steps of Algorithm~\ref{algorithm1}. 



\noindent
\textbf{Step~$1$. Input and Initialization.} 
Each node $v_j$ has two integer\footnote{Since communication is over digital channels, the initial states if not already quantized, are quantized by the nodes. 
Finding the optimal solutions, we aim to find the optimal solution for the case where the initial states are quantized, or are real and were quantized by the nodes.} values $y_j[0], z_j[0] \in \mathbb{Z}$ which represent its initial state (note that $y_j[0] = \alpha_j \rho_j$, and $z_j[0] =\rho_j$). 
Furthermore, each node $v_j$ has knowledge of the diameter of the network $D$ (or an upper bound $D'$). 
During initialization, each node selects a set of probabilities $\{ b_{lj} \ | \ v_{l} \in \mathcal{N}_j^+ \cup \{v_j\} \}$ such that $0 < b_{lj} < 1$ and $\sum_{l=1}^{n} b_{lj} = 1$ (note that $b_{lj} = 0$ for $v_{l} \notin \mathcal{N}_j^+ \cup \{v_j\}$). 
Each value $b_{lj}$, represents the probability for node $v_j$ to transmit towards out-neighbor $v_l \in \mathcal{N}^+_j \cup \{v_j\}$ at any given time step (independently between time steps).  
Furthermore, each node sets its flag, denoted by, $\text{flag}_j$ equal to zero. 
This flag allows node $v_j$ to determine whether it needs to terminate its operation, because the proposed algorithm has reached completion.

\noindent
\textbf{Step~$2$. Synchronously Computing Quantized Optimal Solution.}
At each time step $k$, each node $v_j$ splits $y_j[k]$ into $z_j[k]$ equal integer pieces (with the exception of some pieces whose value might be greater than others by one). 
It chooses one piece with minimum $y$-value and keeps it to itself, and it transmits each of the remaining $z_j[k]-1$ pieces to randomly selected out-neighbors or to itself. 
It receives the values $y_i[k]$ and $z_i[k]$ from its in-neighbors, and sums them with its stored $y_j[k]$ and $z_j[k]$ values, respectively. 
It sets its state variables $z^s_j[k]$, $y^s_j[k]$, equal to its stored $z_j[k]$, $y_j[k]$ values, respectively. 
Then, it updates its state $q^s_j[k]$ to be $y^s_j[k] / z^s_j[k]$.

\noindent
\textbf{Step~$3$. Determining when to Stop.}
Each node $v_j$ has two integer values $M_j$, $m_j$. 
Every $D$ (or $D'$) time steps, the values are set equal to the node's state. 
Then, for the next $D$ (or $D'$) time steps, node $v_j$ executes a $\max$-consensus algorithm \cite{2008:Cortes} with $M_j$ for calculating the maximum state in the network, and a $\min$-consensus algorithm with $m_j$ for calculating the minimum state in the network. 
If at the end of the $D$ (or $D'$) time steps, the maximum state is equal to the minimum state in the network (or if their difference is equal to one), then node $v_j$ knows that the algorithm has converged. 
It calculates the optimal $x^*$ (see \eqref{eq:closedform}), and terminates its operation. 

\begin{varalgorithm}{1}
\caption{Synchronous Quantized Distributed Optimization}
\noindent \textbf{Input:} A strongly connected digraph $\mathcal{G}_d = (\mathcal{V}, \mathcal{E})$ with $n=|\mathcal{V}|$ nodes and $m=|\mathcal{E}|$ edges. 
Each node $v_j\in \mathcal{V}$ has knowledge of $D, y_j[0], z_j[0]$. \\
\textbf{Initialization:} Each node $v_j \in \mathcal{V}$ does the following: 
\begin{list4}
\item[$1)$] Assigns a nonzero probability $b_{lj}$ to each of its outgoing edges $m_{lj}$, where $v_l \in \mathcal{N}^+_j \cup \{v_j\}$, as follows
\begin{align*}
b_{lj} = \left\{ \begin{array}{ll}
         \frac{1}{1 + \mathcal{D}_j^+}, & \mbox{if $l = j$ or $v_{l} \in \mathcal{N}_j^+$,}\\
         0, & \mbox{if $l \neq j$ and $v_{l} \notin \mathcal{N}_j^+$.}\end{array} \right. 
\end{align*} 
\item[$2)$] Sets $\text{flag}_j = 0$, $y_j[0] := 2y_j[0]$, $z_j[0] := 2z_j[0]$. 
\end{list4} 
\textbf{Iteration:} For $k=1,2,\dots$, each node $v_j \in \mathcal{V}$, does the following: 
\begin{list4} 
\item \textbf{while} $\text{flag}_j = 0$ \textbf{then} 
\begin{list4a}
\item[$1)$] \textbf{if} $k \mod D = 1$ \textbf{then} sets $M_j = \lceil y_j[k] / z_j[k] \rceil$, $m_j = \lfloor y_j[k] / z_j[k] \rfloor$; 
\item[$2)$] broadcasts $M_j$, $m_j$ to every $v_{l} \in \mathcal{N}_j^+$; 
\item[$3)$] receives $M_i$, $m_i$ from every $v_{i} \in \mathcal{N}_j^-$; 
\item[$4)$] sets $M_j = \max_{v_{i} \in \mathcal{N}_j^-\cup \{ v_j \}} M_i$, $m_j = \min_{v_{i} \in \mathcal{N}_j^-\cup \{ v_j \}} m_i$; 
\item[$5)$] \textbf{if} $z_j[k] > 1$, \textbf{then} 
\begin{list4a}
\item[$5.1)$] sets $z^s_j[k] = z_j[k]$, $y^s_j[k] = y_j[k]$, 
$
q^s_j[k] = \Bigl \lceil \frac{y^s_j[k]}{z^s_j[k]} \Bigr \rceil \ ;
$
\item[$5.2)$] sets (i) $mas^y[k] = y_j[k]$, $mas^z[k] = z_j[k]$; (ii) $c^{y}_{lj}[k] = 0$, $c^{z}_{lj}[k] = 0$, for every $v_l \in \mathcal{N}^+_j \cup \{v_j\}$; (iii) $\delta = \lfloor mas^y[k] / mas^z[k] \rfloor$, $mas^{rem}[k]= y_j[k] - \delta \ mas^z[k]$;  
\item[$5.3)$] \textbf{while} $mas^z[k] > 1$, \textbf{then} 
\begin{list4a}
\item[$5.3a)$] chooses $v_l \in \mathcal{N}^+_j \cup \{v_j\}$ randomly according to $b_{lj}$; 
\item[$5.3b)$] sets (i) $c^{z}_{lj}[k] := c^{z}_{lj}[k] + 1$, $c^{y}_{lj}[k] := c^{y}_{lj}[k] + \delta$; (ii) $mas^z[k] := mas^z[k] - 1$, $mas^y[k] := mas^y[k] - \delta$. 
\item[$5.3c)$] If $mas^{rem}[k] > 1$, sets $c^{y}_{lj}[k] := c^{y}_{lj}[k] + 1$, $mas^{rem}[k] := mas^{rem}[k]- 1$; 
\end{list4a}
\item[$5.4)$] sets $c^{y}_{jj}[k] := c^{y}_{jj}[k] + mas^y[k]$, $c^{z}_{jj}[k] := c^{z}_{jj}[k] + mas^z[k]$; 
\item[$5.5)$] for every $v_l \in \mathcal{N}^+_j$, if $c^{z}_{lj}[k] > 0$ transmits $c^{y}_{lj}[k]$, $c^{z}_{lj}[k]$ to out-neighbor $v_l$; 
\end{list4a}
\item \textbf{else if} $z_j[k] \leq 1$, sets $c^{y}_{jj}[k] = y_j[k]$, $c^{z}_{jj}[k] = z_j[k]$; 
\item[$6)$] receives $c^{y}_{ji}[k]$, $c^{z}_{ji}[k]$ from $v_i \in \mathcal{N}_j^-$ and sets 
\begin{equation}\label{no_del_eq_y_no_oscil}
y_j[k+1] = c^{y}_{jj}[k] + \sum_{i=1}^{n}  w_{ji}[k] \ c^{y}_{ji}[k] ,
\end{equation} 
\begin{equation}\label{no_del_eq_z_no_oscil}
z_j[k+1] = c^{z}_{jj}[k] + \sum_{i=1}^{n} w_{ji}[k] \ c^{z}_{ji}[k] ,
\end{equation}
where $w_{ji}[k] = 1$ if node $v_j$ receives $c^{y}_{ji}[k]$, $c^{z}_{ji}[k]$ from $v_i \in \mathcal{N}_j^-$ at iteration $k$ (otherwise $w_{ji}[k] = 0$); 
\item[$7)$] \textbf{if} $k \mod D = 0$ \textbf{then}, \textbf{if} $M_j - m_j \leq 1$ 
\textbf{then} sets $q^s_j[k] = m_j$, $x_j^* = x^* = y_j[0] / q^s_j[k]$
and $\text{flag}_j = 1$. 
\end{list4a}
\end{list4}
\textbf{Output:} \eqref{eq:closedform} holds for every $v_j \in \mathcal{V}$. 
\label{algorithm1}
\end{varalgorithm}



Next, we show that, during the operation of Algorithm~\ref{algorithm1}, each node $v_j$ is able to calculate $x^*$ after a finite number of time steps. 


\begin{theorem}
\label{thm:main}
Consider a strongly connected digraph $\mathcal{G}_d = (\mathcal{V}, \mathcal{E})$ with $n=|\mathcal{V}|$ nodes and $m=|\mathcal{E}|$ edges and $z_j[0]$, $y_j[0]$ for every node $v_j \in \mathcal{V}$ at time step $k=0$. 
Suppose that each node $v_j \in \mathcal{V}$ follows the Initialization and Iteration steps as described in Algorithm~\ref{algorithm1}.
Each node $v_j$ is able to calculate the optimal $x^*$ shown in \eqref{eq:closedform} after a finite number of time steps and terminate its operation after calculating $x^*$.
\end{theorem}

\begin{proof}
See Appendix~\ref{appendix:A}.
\end{proof}

\begin{remark}
Note that Algorithm~\ref{algorithm1} is based on similar principles as the algorithm presented in \cite{2015:Cady}, which executes the ratio-consensus algorithm \cite{2010:christoforos} along with $\min-$ and $\max-$consensus iterations \cite{2008:Cortes}. 
Overall, \cite{2015:Cady} allowed the nodes in the network to calculate the {\em real} average of their initial states and then terminate their operation according to a distributed stopping criterion. 
However, Algorithm~\ref{algorithm1} has significant differences due to its quantized nature. 
These differences mainly focus on (i) the underlying process for calculating the quantized average of the initial states via the exchange of quantized messages, and (ii) the distributed stopping mechanism designed explicitly for quantized information exchange algorithms.  
It is also interesting to note that Algorithm~\ref{algorithm1} is substantially different than the algorithm in \cite{2021:Rikos_Hadj_Johan}. 
These differences mainly focus on (i) the final calculation of the optimal solution, and (ii) utilization of a distributed stopping mechanism to cease transmissions. 
\end{remark}

\subsection{Asynchronous Quantized Distributed Solution}\label{asynch_extension_subse}

We now focus on the case where nodes operate in an asynchronous fashion. 
We first make the following assumption, and then we present an asynchronous version of the distributed stopping mechanism in Section~\ref{alg_code}. 

\begin{assum}\label{assum_upper_bound_process}
The number of time steps required for a node $v_j$ to process the information received from its in-neighbors is upper bounded by $\mathcal{B}$. 
\end{assum}

Assumption~\ref{assum_upper_bound_process} is necessary for the operation of the asynchronous version of the max/min consensus algorithm. 
Specifically, if we have a bound on the number of time steps required for a node $v_j$ to process the information received from its in-neighbors, we can ensure that each node $v_j$ can still determine whether convergence has been achieved or not, even when operating asynchronously. 

\vspace{.2cm}

\noindent

\textbf{Asynchronous max/min - Consensus.} 
In the asynchronous version of max/min consensus, the update rule for every node $v_j \in \mathcal{V}$ is \cite{2013:Giannini}: 
\begin{align}\label{asynch_max_operation_eq}
x_j[k + \theta_j[k]] = \max_{v_{i}\in \mathcal{N}_j^{-} \cup \{v_{j}\}}\{ x_i[k + \theta_{ji}[k]] \}, 
\end{align}
where $\theta_j[k]$ is the update instance of node $v_j$, $x_i[k + \theta_{ji}[k]]$ are the states of the in-neighbors $v_{i}\in \mathcal{N}_j^{-} \cup \{v_{j}\}$ during the time instant of $v_j$'s update, $\theta_{ji}[k]$ are the asynchronous state updates of the in-neighbors of node $v_j$ that occur between two consecutive updates of node $v_j$'s state. 
The asynchronous version of the max/min consensus in \eqref{asynch_max_operation_eq} converges to the maximum value among all nodes in a finite number of steps $s' \leq D \mathcal{B}$ (as shown in \cite{2013:Giannini}), where $D$ is the diameter of the network, and $\mathcal{B}$ is the upper bound on the number of time steps required for a node $v_j$ to process the information received from its in-neighbors. 

We now describe the main steps of Algorithm~\ref{algorithm2}. 

\noindent
\textbf{Step~$1$. Input and Initialization.} 
Input and initialization steps are the same as Algorithm~\ref{algorithm1}.

\noindent
\textbf{Step~$2$. Asynchronously Computing Quantized Optimal Solution.}
At each time step $k$, nodes adjust the operation of the underlying quantized average consensus for the case where every node suffers arbitrary processing delays. 
Specifically, each node $v_j$ splits $y_j[k]$ into $z_j[k]$ equal integer pieces (with the exception of some pieces whose value might be greater than others by one).  
This processing requires a number of time steps upper bounded by $\mathcal{B}$. 
Then, it keeps to itself the piece with minimum $y$-value, and transmits each of the remaining $z_j[k]-1$ pieces to randomly selected out-neighbors or to itself. 
It receives the values $y_i[k]$ and $z_i[k]$ from its in-neighbors, and sums them with its stored $y_j[k]$ and $z_j[k]$ values. 
It sets its state variables $z^s_j[k]$, $y^s_j[k]$, equal to its stored $z_j[k]$, $y_j[k]$ values, respectively. 
Then, it updates its state $q^s_j[k]$ to be $y^s_j[k] / z^s_j[k]$.


\noindent
\textbf{Step~$3$. Asynchronously Determining when to Stop.}
Each node $v_j$ has two integer values $M_j$, $m_j$. 
Also each node has knowledge of $\mathcal{B}$. 
The main idea is that each node prolongs for $D \mathcal{B}$ time steps the max/min consensus operation in order for every node to participate in the max/min consensus operation (despite its random processing delays which are upper bounded by $\mathcal{B}$). 
Specifically, every $D \mathcal{B}$ (or $D' \mathcal{B}$) time steps, the values are set equal to the node's state. 
Then, for the subsequent $D \mathcal{B}$ (or $D' \mathcal{B}$) time steps, node $v_j$ executes an asynchronous $\max$-consensus algorithm with $M_j$ for calculating the maximum state in the network, and an asynchronous $\min$-consensus algorithm with $m_j$ for calculating the minimum state in the network. 
If the maximum state is equal to the minimum state in the network (or their difference is equal to one), then node $v_j$ knows that the algorithm has converged.


We now analyze the operation and the convergence rate of Algorithm~\ref{algorithm2}.

\begin{varalgorithm}{2}
\caption{Asynchronous Quantized Distributed Optimization}
\noindent \textbf{Input:} A strongly connected digraph $\mathcal{G}_d = (\mathcal{V}, \mathcal{E})$ with $n=|\mathcal{V}|$ nodes and $m=|\mathcal{E}|$ edges. 
Each node $v_j\in \mathcal{V}$ has knowledge of $D, \mathcal{B}, y_j[0], z_j[0]$. \\
\textbf{Initialization:} Same as Algorithm~\ref{algorithm1}. \\
\textbf{Iteration:} For $k=1,2,\dots$, each node $v_j \in \mathcal{V}$, does the following: 
\begin{list4} 
\item \textbf{while} $\text{flag}_j = 0$ \textbf{then} 
\begin{list4a}
\item[$1)$] \textbf{if} $k \mod (D \mathcal{B}) = 1$ \textbf{then} sets $M_j = \lceil y_j[k] / z_j[k] \rceil$, $m_j = \lfloor y_j[k] / z_j[k] \rfloor$; 
\item[$2 - 5)$] same as Algorithm~\ref{algorithm1}; 
\item[$6)$] receives $c^{y}_{ji}[k]$, $c^{z}_{ji}[k]$ from $v_i \in \mathcal{N}_j^-$ and sets 
\begin{equation}\label{no_del_eq_y_no_oscil_asynch}
y_j[k+1] = c^{y}_{jj}[k] + \sum_{i=1}^{n} \sum_{r=0}^{\mathcal{B}}  w_{k-r,ji}[r] \ c^{y}_{ji}[k-r] ,
\end{equation} 
\begin{equation}\label{no_del_eq_z_no_oscil_asynch}
z_j[k+1] = c^{z}_{jj}[k] + \sum_{i=1}^{n} \sum_{r=0}^{\mathcal{B}} w_{k-r,ji}[r] \ c^{z}_{ji}[k-r] ,
\end{equation}
where $w_{k-r,ji}[r] = 1$ when the required processing time of node $v_i$ is equal to $r$ at time step $k-r$, so that node $v_j$ receives $c^{y}_{ji}[k]$, $c^{z}_{ji}[k]$ from $v_i$ at time step $k$ (otherwise $w_{k-r,ji}[r] = 0$ and $v_j$ receives no message at time step $k$ from $v_i$); 
\item[$7)$] \textbf{if} $k \mod (D \mathcal{B}) = 0$ \textbf{then}, \textbf{if} $M_j - m_j \leq 1$ 
\textbf{then} sets $q^s_j[k] = m_j$, $x_j^* = x^* = y_j[0] / q^s_j[k]$
and $\text{flag}_j = 1$. 
\end{list4a}
\end{list4}
\textbf{Output:} \eqref{cond:balance} holds for every $v_j \in \mathcal{V}$. 
\label{algorithm2}
\end{varalgorithm}



\begin{theorem}
\label{thm:main_2} 
Consider a strongly connected digraph $\mathcal{G}_d = (\mathcal{V}, \mathcal{E})$ with $n=|\mathcal{V}|$ nodes and $m=|\mathcal{E}|$ edges, and initial values $z_j[0]$, $y_j[0]$
for every node $v_j \in \mathcal{V}$ at time step $k=0$. 
Suppose that each node $v_j \in \mathcal{V}$ follows the Initialization and Iteration steps as described in Algorithm~\ref{algorithm2}.
Each node $v_j$ is able to (i) calculate the optimal $x^*$ shown in \eqref{eq:closedform} after a finite number of time steps and (ii) after calculating $x^*$, terminate its operation.
\end{theorem}

\begin{proof}
See Appendix~\ref{appendix:B}.
\end{proof}

\section{Applications}\label{sec:applications}



We now present applications and comparisons of Algorithm~\ref{algorithm1} and Algorithm~\ref{algorithm2}. 
We focus on the following applications: (i) a set of server nodes in a data center that aim to balance their CPU utilization by deciding how to allocate a set of tasks to CPU resources \cite{2020:Themis_Kalyvianaki, 2021:Rikos_Grammenos_Kalyvianaki_Hadj_Themis_Johan_CPU}, and (ii) a set of processing nodes in a federated learning system that aim to aggregate the parameters of their locally trained models in order to calculate the parameters of a global model \cite{2020:Miao, 2020:Virginia, 2019:Hu_Wang}. 
Finally, we compare Algorithm~\ref{algorithm1} and Algorithm~\ref{algorithm2} with other algorithms in the current literature. 

\subsection{Quantized Task Scheduling}\label{appl_cpu}


Resource management in data centers is the procedure of allocating a set of tasks efficiently to CPU resources 
such that certain performance objectives can be satisfied. 
Resource allocation can be cast as an optimization problem. 
However, solving it is challenging due to the scale, heterogeneity, and dynamic nature of modern computer networks. 
More specifically, in our setting we have that a set of server nodes operate over a large-scale network (i.e., a data center). 
Due to device heterogeneity, server nodes may have different processing capabilities. 
Task scheduling aims to balance CPU utilization across server nodes by carefully deciding how to allocate tasks to CPU resources in a distributed fashion \cite{2020:Themis_Kalyvianaki}. 

In what follows we describe the task modeling and optimization problem for CPU scheduling. 
Note that these are borrowed from \cite{2020:Themis_Kalyvianaki}. 

\textbf{Task Modelling (\hspace{-0.001cm}\cite{2020:Themis_Kalyvianaki}).}
A job is defined as a group of tasks, and $\mathcal{J}$ denotes the set of all jobs to be scheduled. 
Each job $b_{j} \in \mathcal{J}$, $j\in\{1,\ldots, |\mathcal{J}| \}$, requires $\rho_{j}$ cycles to be executed.
The estimated amount of resources (i.e., CPU cycles) needed for each job is assumed to be known before the optimization starts. 
A job task could require resources ranging from 1 to $\rho_j$ cycles, and the total sum of resources for all tasks of the same job is equal to $\rho_j$ cycles.
The total task workload due to the jobs arriving at node $v_i$ is denoted by $l_i$. 
The time horizon $T_{h}$ is defined as the time period for which the optimization is considering the jobs to be running on the server nodes, before the next optimization decides the next allocation of resources. 
Hence, in this setting, the CPU capacity of each node, considered during the optimization, is computed as
$\pi_i^{\max} \coloneqq c_i T_h$, where $c_i$ is the sum of all clock rate frequencies of all processing cores of node $v_{i}$ given in cycles/second. 
The CPU availability for node $v_{i}$ at optimization step $m$ (i.e., at time $mT_{h}$) is given by $\pi_i^{\mathrm{avail}}[m] \coloneqq \pi_i^{\max} - u_{i}[m]$, where $u_{i}[m]$ is the number of unavailable/occupied cycles due to predicted or known utilization from already running tasks on the server over the time horizon $T_{h}$ at step $m$. 
Note here that all the above quantities are discrete values. 
Thus, they can be represented by integer values. 

\begin{assum}\label{res_less_capacity}
    We assume that the time horizon is chosen such that the total amount of resources demanded at a specific optimization step $m$, denoted by $\rho[m] \coloneqq 
    \sum_{j=1}^{n} \rho_{j}[m]$,
    is smaller than the total capacity of the network available, given by $\pi^{\mathrm{avail}}[m] \coloneqq \sum_{i=1}^{n} \pi_i^{\mathrm{avail}}[m]$, i.e., $\rho[m] \leq \pi^{\mathrm{avail}}[m]$. 
\end{assum}

Assumption~\ref{res_less_capacity} indicates that there is no more demand than the available resources. 
This assumption is realistic, since the time horizon $T_h$ can be chosen appropriately to fulfill the requirement. 
In case this assumption is violated, the solution will be that all resources are being used and some tasks will not be scheduled, due to lack of resources.
Please note that handling this case is out of the scope of this paper. 

\textbf{Optimization Problem (\hspace{-0.001cm}\cite{2020:Themis_Kalyvianaki}).}
Server nodes require to calculate the optimal solution at every optimization step $m$ via a distributed coordination algorithm which relies on the exchange of quantized values and converges after a finite number of time steps. 
Specifically, all nodes aim to balance their CPU utilization (i.e., the same percentage of capacity) during the execution of the tasks, i.e., 
\begin{align}\label{cond:balance}
\frac{w_i^*[m] +u_{i}[m]}{\pi_i^{\max}} &= \frac{w_j^*[m] +u_{j}[m]}{\pi_j^{\max}} \\
&= \frac{\rho[m] + u_{\mathrm{tot}}[m]}{\pi^{\max}}, \ \forall v_{i}, v_{j} \in \mathcal{V}, \nonumber
\end{align}
where $w_i^*[m]$ is the \emph{optimal} task workload to be added to server node $v_{i}$ at optimization step $m$, $\pi^{\max} \coloneqq \sum_{i=1}^{n} \pi_i^{\max}$ and $u_{\mathrm{tot}}[m] = \sum_{i=1}^{n} u_{i}[m]$. 
To achieve the requirement set in \eqref{cond:balance}, we need the solution (according to \eqref{eq:closedform}) to be \cite{2020:Themis_Kalyvianaki} 
\begin{align}\label{eq:closedform1}
x^* =  \frac{\sum_{i=1}^{n} \pi_i^{\max} \frac{\rho_{i}+u_{i}}{\pi_i^{\max}}}{\sum_{i=1}^{n} \pi_i^{\max}} = \frac{\rho + u_{\mathrm{tot}}}{\pi^{\max}}.
\end{align}
Hence, we modify \eqref{local_cost_functions} accordingly. 
Then, the cost function $f_i(z)$ in \eqref{local_cost_functions} is given by
\begin{align}\label{eq:fiz}
f_i(z) = \frac{1}{2}\pi_i^{\max} \left(z- \frac{\rho_{i}+u_{i}}{\pi_i^{\max}} \right)^2.
\end{align}
In other words, each node computes its proportion of task workload and from that it computes the task workload $w_i^*$ to receive, i.e.,
 \begin{align}\label{eq:optimal_workload}
w_i^*  = \frac{\rho + u_{\mathrm{tot}}}{\pi^{\max}} \pi_i^{\max} - u_{i}.
\end{align} 


\textbf{Application of Algorithm~\ref{algorithm1} and Algorithm~\ref{algorithm2}.}
During the task scheduling problem each node $v_j$ aims to (i) calculate the optimal required task workload $w^*_j$ (shown in \eqref{eq:optimal_workload}) after a finite number of time steps, and (ii) terminate its operation after calculating $w_j^*$. 
In order to solve the above task scheduling problem, Algorithm~\ref{algorithm1} and Algorithm~\ref{algorithm2} need to be modified as follows: each node $v_j$ needs to (i) have knowledge of $D, \mathcal{B}, l_j, u_j, \pi_j^{\max} \in \mathbb{Z}$ ($\mathcal{B}$ is required only for Algorithm~\ref{algorithm2}), and (ii) initialize $z_j[0] := l_j + u_j$, $y_j[0] = \pi_j^{\max}$.

\begin{remark}
Note that the main difference of the operation of Algorithm~\ref{algorithm1} compared to the algorithm presented in \cite{2021:Rikos_Grammenos_Kalyvianaki_Hadj_Themis_Johan_CPU} is that each server node $v_j \in \mathcal{V}$ does \textit{not} need knowledge of an upper bound $\pi^{\mathrm{upper}}$ regarding the total capacity of the network $\pi^{\max}$ (i.e., $\pi^{\mathrm{upper}} \geq \pi^{\max}$, where $\pi^{\max} \coloneqq \sum_{j=1}^{n} \pi_j^{\max}$). 
Specifically, each node does not need to multiply its initial value $y_j[0]$ with $\pi^{\mathrm{upper}}$ so that $y_j[0] > z_j[0]$, since it is already guaranteed that $y_j[0] > z_j[0]$ (which is necessary during the operation of our algorithm, so that each node $v_j$ is able to split $y_j[k]$ into $z_j[k]$ equal integer pieces (or with maximum difference between them equal to $1$) at each time step $k \in \mathbb{N}$). 
Please note that this relaxation of requirements does not affect the operation of the algorithm and its fast convergence speed, as it will be shown later.  
\end{remark}

\textbf{Numerical Evaluation over a Small Network.} 
We now present simulation results to illustrate the behavior of our proposed distributed algorithms over a random graph of $20$ nodes and show how the states of the nodes converge. 
To the best of our knowledge, this is the first work that tries to tackle the problem of converging using quantized values at that scale while also providing a thorough evaluation accompanied with strong theoretical guarantees.
To foster reproducibility, the code, datasets, and experiments are made publicly available.\footnote{\url{https://github.com/andylamp/federated-quantized-ratio-consensus}}
We show the evolution of the nodes' states against the number of iterations during the operation of Algorithm~\ref{algorithm1} and Algorithm~\ref{algorithm2}.
The network in this example comprises $20$ nodes and was randomly generated (an edge between a pair of nodes exists with probability $0.5$).
This process resulted in a digraph that had a diameter equal to $2$. 
Small digraph diameters are indicative on data-center topologies and are normally preferred due to their locality and the benefit of having few hops between each node~\cite{besta2014slim}.
The task workload $l_j$ of each node $v_j$ was generated randomly using an integer random distribution uniform in the range $[1, 100]$. 
The node capacities $\pi_j^{\max}$ in this experiment were set to either $100$ or $300$ for even and odd node numbers respectively.
Our simulation results are shown in Fig.~\ref{eval-small-example}, which depicts the load per node according to its processing capacity during the operation of Algorithms~\ref{algorithm1} and~\ref{algorithm2}.
Specifically, we show $q_j^{s'}[k] = y_j[0] / q^s_j[k]$ for every node. 
We can see that Algorithm~\ref{algorithm1} converges monotonically within a few iterations (i.e., after only $8$ iterations) without being affected by value oscillations or ambiguities. 
Furthermore, we can see that Algorithm~\ref{algorithm2} requires more iterations due to the number of time steps that each node requires to process information (which has an upper bound equal to $5$ i.e., $\mathcal{B} = 5$). 
We can see that Algorithm~\ref{algorithm2} requires more time steps to converge compared to Algorithm~\ref{algorithm1}; this is mainly due to processing delays leading to delayed execution of the algorithm's iteration steps (instead of instant execution for the case of no processing delays) \cite{2014:ThemisTAC}. 
However, note here that in most cases asynchronous algorithms are able to deal with processing delays more efficiently than synchronous ones which perform poorly in heterogeneous environments \cite{2020:Assran_Rabbat, 2018:Lian_Liu, 2018:Dutta_Nagpurkar}. 
This means that for the case where every node suffers from processing delays, synchronous algorithms suffer from delays while waiting for the slow processing nodes. 
In our case this means that every node should wait for $5$ time steps before executing a synchronized iteration, and Algorithm~\ref{algorithm1} would require $45$ time steps for convergence. 
On the other hand, Algorithm~\ref{algorithm2} is able to achieve faster convergence (since it is not sensitive to issues like slow computing nodes) and uses computational resources more efficiently than synchronous approaches \cite{2020:Assran_Rabbat}.

\begin{figure}[t]
    \centering
    \includegraphics[width=8cm]{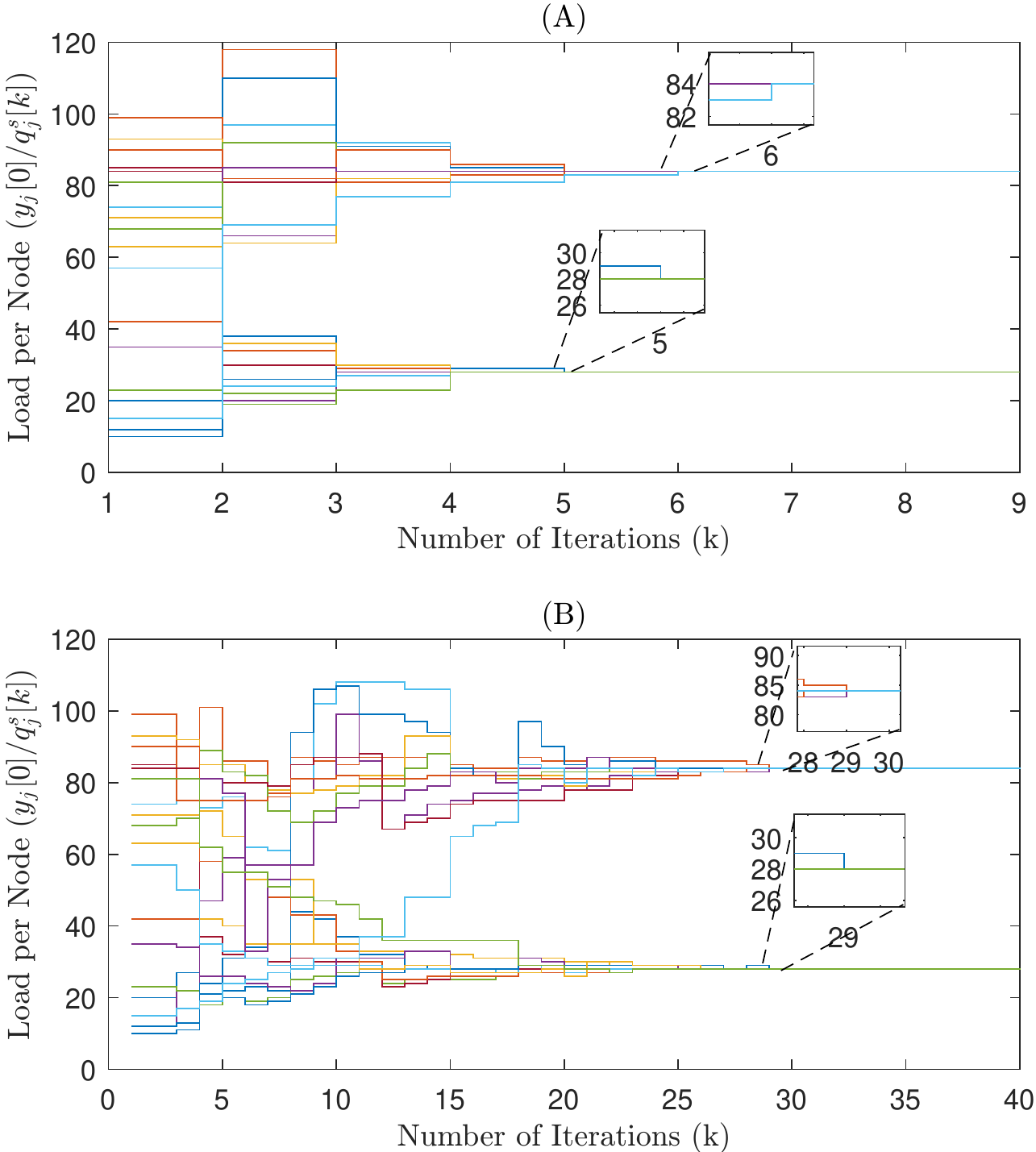}
    \caption{Execution of Algorithm~\ref{algorithm1} (A), and Algorithm~\ref{algorithm2} with $\mathcal{B} = 5$ (B), over a random network comprised of $20$ nodes having a diameter equal to $2$. 
    Note the different scale of the x-axis; the convergence of the synchronous algorithm is about five times faster than the asynchronous.
    }
    \label{eval-small-example}
\end{figure}

\textbf{Numerical Evaluation.}
We now present a more quantitative analysis over a larger set of network sizes, which would be more applicable to practical deployments, such as in modern data-centers. 
In Fig.~\ref{converge-stats}, we evaluate Algorithm~\ref{algorithm2} on networks sized from $50$ nodes up to $3000$ nodes. 
We show the number of required iterations for convergence for different network sizes and different values of the upper bound $\mathcal{B}$ on the number of time steps required for a node to process information from $5$ to $30$. 
The topologies are randomly generated and result in digraphs that have a diameter from $2$ to $10$.
We evaluated each network size across $3000$ trials and the aggregated values were averaged out before plotting.
It is interesting to see that for $\mathcal{B} = 5$, Algorithm~\ref{algorithm2} required less than $250$ iterations to converge for every network size. 
Also, for $\mathcal{B} = 5$, we can see that, as the network size increases, the required iterations for convergence decrease. 
Another interesting observation is that as the value of the upper bound $\mathcal{B}$ on the number of time steps required for a node to process information increases, the number of required time steps for convergence increases linearly. 
This means, for example, that (i) for $\mathcal{B} = 10$, Algorithm~\ref{algorithm2} required less than $280$ iterations to converge for every network size, and (ii) for $\mathcal{B} = 15$, Algorithm~\ref{algorithm2} required less than $350$ iterations to converge for every network size.

\begin{figure}[t]
    \centering
    \includegraphics[width=8cm]{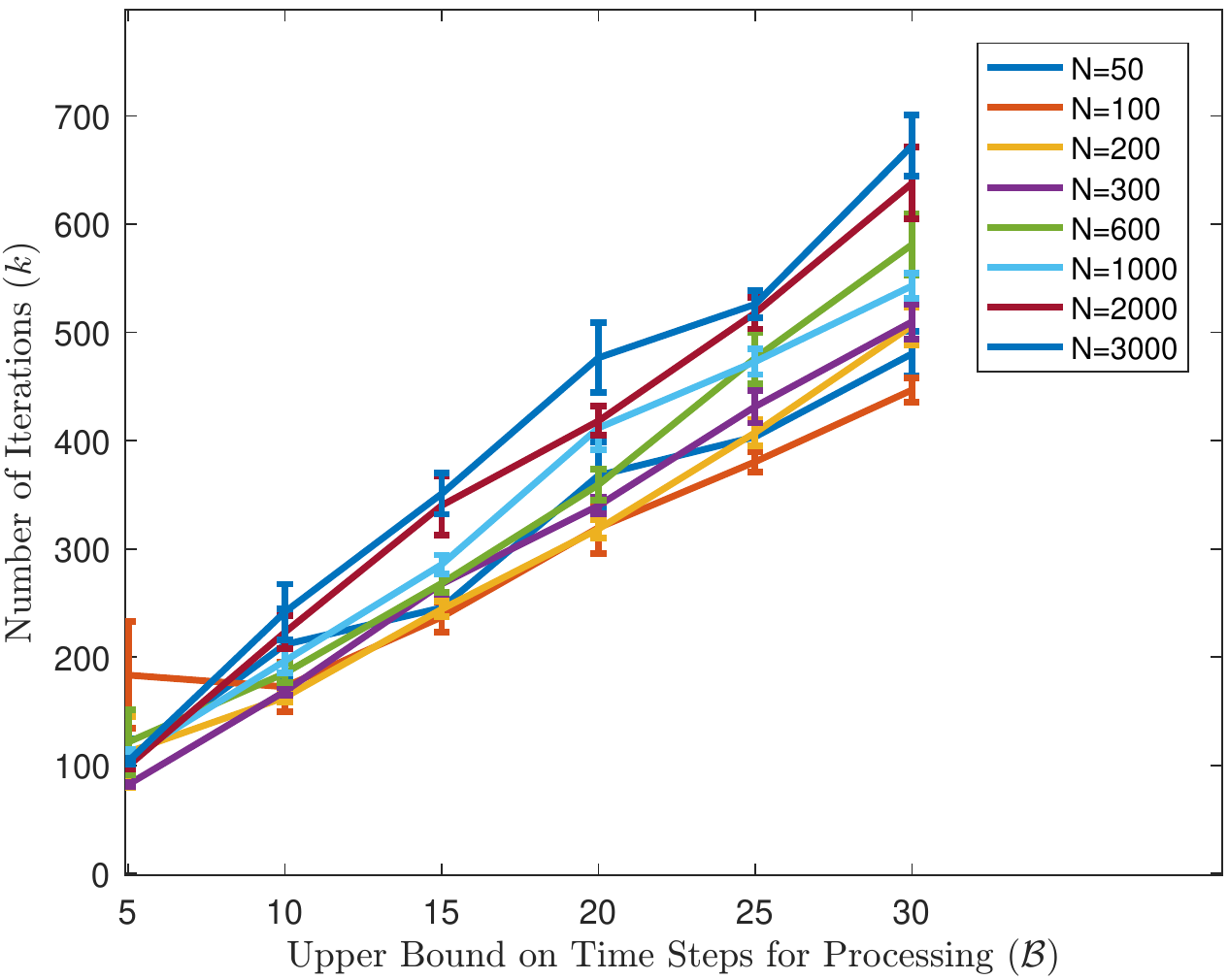}
    \caption{Required iterations for convergence of Algorithm~\ref{algorithm2} against upper bound on required time steps for processing $\mathcal{B}$, over different network sizes along with their error bars, where each network size is evaluated across $3000$ trials and the aggregated values were averaged out before plotting. 
    }
    \label{converge-stats}
\end{figure}

\subsection{Quantized Global Model Aggregation}\label{appl_federated}

In Federated Learning we aim to learn the parameters of a specific statistical model from data stored on thousands (or millions) of remote devices. 
Most current approaches consider the existence of a central server which collects and aggregates the computed local models from every node in the network \cite{2018:Smith_Jaggi}.
However, communication overhead during each iteration becomes a major bottleneck as the model size gets large and the computed models increase in dimension. 
Therefore, aggregating the local models in a centralized manner is not an ideal approach due to inefficient operation and practical limitations.  
In our setting, we consider a set of remote devices (i.e., processing units or nodes) over a network. 
Each node has a stored local dataset, and uses it to calculate the local parameters of the statistical model. 
Global model aggregation is the procedure of aggregating the set of local parameters of every node in the network, in a distributed fashion. 
This procedure aims to calculate the global parameters of the statistical model. 
Furthermore, global model aggregation needs to be performed in a communication efficient manner, since communication efficiency is a critical bottleneck in Federated Learning systems \cite{2019:Roselander}. 
Thus, in our case, nodes need to transmit quantized values in order to achieve more efficient usage of communication resources.


\textbf{Global and Local Models.} 
For each node $v_j$, its stored dataset is denoted as $r_j$ and its size as $| \mathcal{R}_j |$. 
Each node trains a local model with its own dataset, and then transmits the local model parameters (e.g., gradients) in the network for aggregation. 
The local model parameters of each node $v_j$ at aggregation step $m$ are denoted as $\mathcal{W}_j[m]$. 
Furthermore, the global model parameters at aggregation step $m$ are denoted as $\mathcal{W}[m]$ and are calculated after aggregating the local model parameters $\mathcal{W}_j[m]$ of every node $v_j$. 
Note that in our case, the communication channels are bandlimited. 
As a result, the parameters of the local model $\mathcal{W}_j[m]$ for every node $v_j$, and the global model $\mathcal{W}[m]$ are quantized values. 
In this scenario we represent them as integer values in order to present an illustrative application of our proposed algorithms.

\textbf{Optimization Problem.}
In a federated learning system, computing nodes require to calculate the parameters of the global model by aggregating the parameters of their local model. 
This is done at aggregation step $m$ via a distributed coordination algorithm which relies on the exchange of quantized values and converges after a finite number of time steps. 
More specifically, each node $v_j$ aims to calculate the global model parameters $\mathcal{W}[m]$ at aggregation step $m$ defined as \cite{2020:Miao, 2019:Hu_Wang} 
\begin{equation}\label{cond:global_model_parameters}
    \mathcal{W}[m] = \frac{\sum_{j=1}^{n} | \mathcal{R}_j | \mathcal{W}_j[m] }{ \sum_{j=1}^{n} | \mathcal{R}_j | } .
\end{equation}
For simplicity of exposition, and since we consider a single aggregation step, we drop index $m$.
To achieve the requirement in \eqref{cond:global_model_parameters}, we need to modify \eqref{local_cost_functions} to be 
\begin{equation}\label{local_cost_functions_federated}
    f_j(z) = \dfrac{1}{2} | \mathcal{R}_j | (z - \mathcal{W}_j)^2 . 
\end{equation}
This means that the closed form solution of \eqref{eq:closedform} becomes 
\begin{align}\label{eq:closedform_federated}
x^* =  \frac{\sum_{j=1}^{n} | \mathcal{R}_j | \mathcal{W}_j}{\sum_{j=1}^{n} | \mathcal{R}_j |}. 
\end{align}
In other words, each node computes the parameters of the global model $x^*$ shown in \eqref{eq:closedform_federated}.

\textbf{Application of Algorithm~\ref{algorithm1} and Algorithm~\ref{algorithm2}.}
During the global model aggregation problem, each node $v_j$ aims to (i) calculate the global model parameters $x^*$ (shown in \eqref{eq:closedform_federated}) after a finite number of time steps, and (ii) terminate its operation after calculating $x^*$. 
In order to solve the global model aggregation problem, Algorithm~\ref{algorithm1} and Algorithm~\ref{algorithm2} need to be modified as follows: each node $v_j$ needs to (i) have knowledge of $D, \mathcal{B}, | \mathcal{R}_j |, \mathcal{W}_j$ ($\mathcal{B}$ is required only for Algorithm~\ref{algorithm2}), and (ii) initialize $z_j[0] := | \mathcal{R}_j |$, $y_j[0] := \mathcal{W}_j$.

\textbf{Numerical Evaluation.} 
We illustrate the behavior of our proposed distributed algorithms over a random graph of $20$ nodes with simulations and show how the states of the nodes converge. 
The network comprised $20$ nodes and was randomly generated (an edge between a pair of nodes was created with probability $0.5$). 
This resulted in a digraph with diameter equal to $3$. 
The size $| \mathcal{R}_j |$ of the stored dataset $r_j$ for each node was generated randomly using an integer random distribution, uniform in the range $[10, 100]$. 
The parameters of the local models $\mathcal{W}_j$ for each node $v_j$ were generated randomly using an integer random distribution uniform in the range $[1000, 100000]$. 
Our simulation results are presented in Fig.~\ref{20nodes}, which depicts the local model parameters (initial state) of each node and the calculation of the global model parameters (final state of each node) in finite time. 
During Algorithm~\ref{algorithm1} we can see that each node converges monotonically after $9$ iterations and then terminates its operation. 
During Algorithm~\ref{algorithm2} we can see that convergence keeps its monotonic nature, but each node converges after $69$ iterations due to the required time for information processing (which is equal to $5$). 

\begin{figure}[t]
    \centering
    \includegraphics[width=8cm]{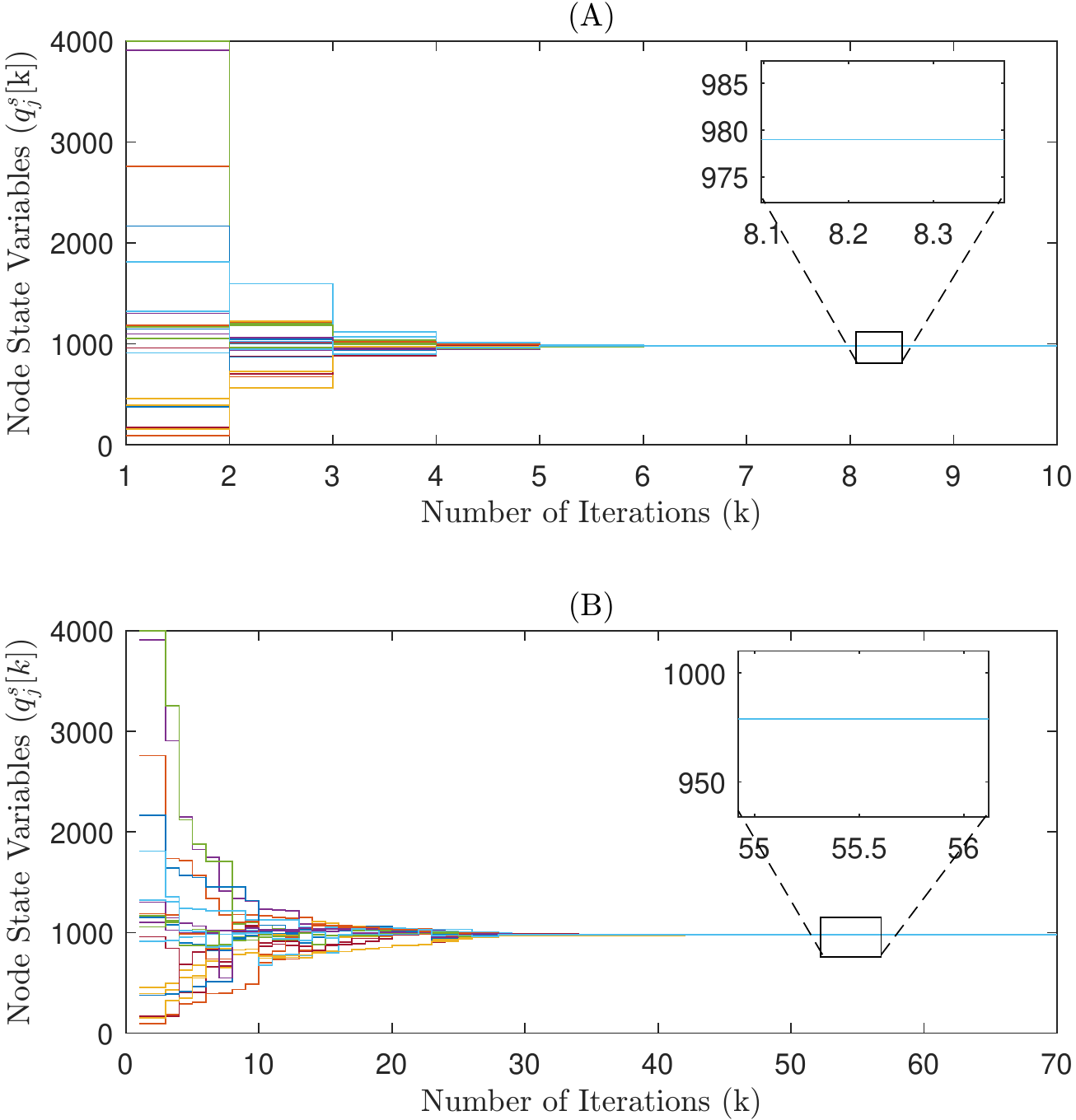}
    \caption{Execution of Algorithm~\ref{algorithm1} (A), and Algorithm~\ref{algorithm2} with $\mathcal{B} = 5$ (B), over a random network comprised of $20$ nodes having a diameter equal to $3$. 
    }
    \label{20nodes}
\end{figure}

\subsection{Comparison with Current Literature}\label{compar_sec} 


We now compare the performance of Algorithm~\ref{algorithm1} against existing algorithms over static strongly connected directed networks of $20$ nodes. 
Specifically, we show the normalized error $e[k]$ defined as 
\begin{equation}\label{eq:distance_optimal}
    e[k] = \sqrt{ \frac{ \sum_{j=1}^n ((q_j[k])^{-1} - x^*)^2 }{ \sum_{j=1}^n ((q_j[0])^{-1} - x^*)^2 } } , 
\end{equation}
where $x^*$ is defined in \eqref{eq:closedform}. 
The error $e[k]$ was evaluated and averaged across $20$ trials. 
In Fig.~\ref{comparisons_literature_optim}, we can see that Algorithm~\ref{algorithm1} is among the fastest algorithms in the literature outperformed only by \cite{2020:Themis_Kalyvianaki}. 
Our algorithm operates with quantized values which influence the convergence rate (thus outperformed by \cite{2020:Themis_Kalyvianaki}), but admits finite time convergence to the \textit{proximity of the} optimal solution, with the distance from the exact solution depending on the quantization level. 
Most algorithms in the literature assume that the messages exchanged among nodes in the network are real numbers and admit asymptotic convergence within some error \cite{2012:rabbat_pushsum, 2020:Themis_Kalyvianaki, 2009:Nedic_Optim, 2018:Xie, 2018:Khan_addopt, 2018:Khan_AB, 2021:Nedic_PushPull}. 
In \cite{2020:Doostmohammadian_Charalambous} the exchanged messages are quantized but it still exhibits asymptotic convergence. 
Furthermore, an additional advantage of our algorithm is that its operation does not rely on a set of weights on the digraph links\footnote{The algorithms in \cite{2009:Nedic_Optim, 2018:Xie, 2020:Doostmohammadian_Charalambous} require the underlying graph to be undirected. 
For this reason, in Fig.~\ref{comparisons_literature_optim}, for \cite{2009:Nedic_Optim, 2018:Xie, 2020:Doostmohammadian_Charalambous}, we make the randomly generated underlying digraphs undirected (by enforcing that if $(v_j, v_i) \in \mathcal{E}$ then also $(v_i, v_j) \in \mathcal{E}$). 
For the algorithms in \cite{2012:rabbat_pushsum, 2020:Themis_Kalyvianaki, 2018:Khan_addopt, 2018:Khan_AB, 2021:Nedic_PushPull}, the randomly generated underlying graph is generally directed.} that form a double stochastic matrix, unlike \cite{2009:Nedic_Optim, 2018:Xie}.

\begin{figure}[t]
    \centering
    \includegraphics[width=8cm]{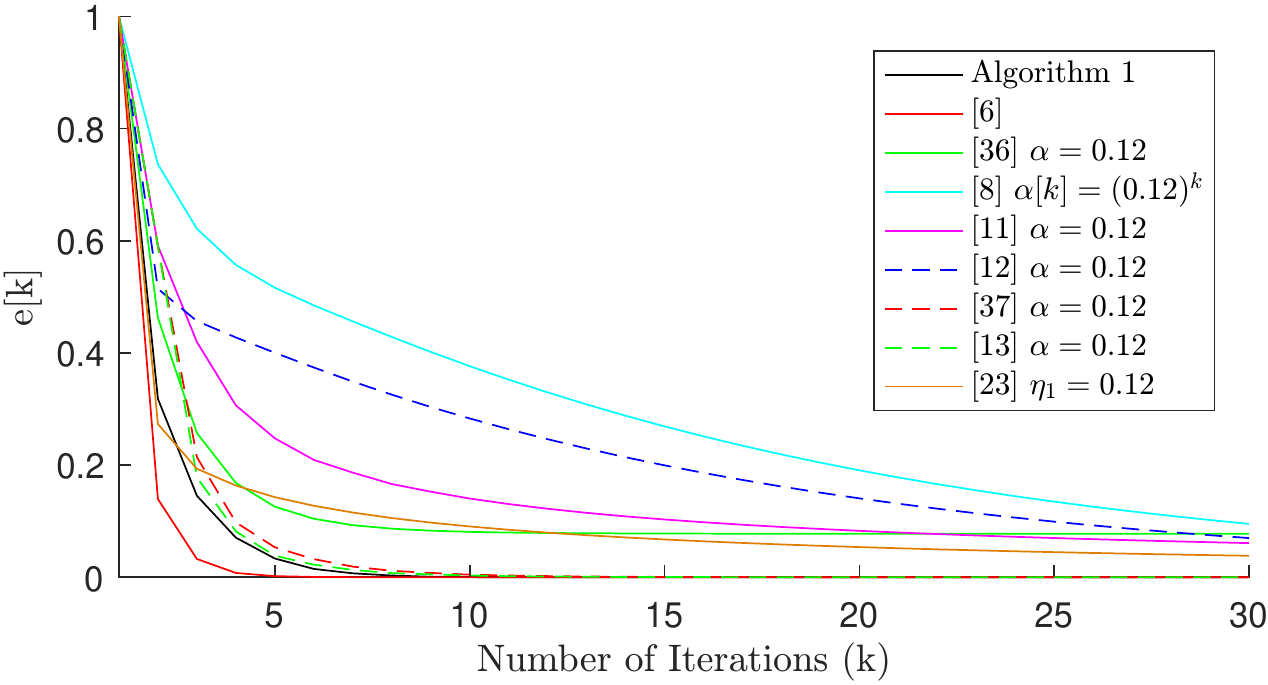}
    \caption{
    Normalized error $e[k]$ (defined in \eqref{eq:distance_optimal}) 
for Algorithm~\ref{algorithm1}, and the algorithms in \cite{2012:rabbat_pushsum, 2020:Themis_Kalyvianaki, 2009:Nedic_Optim, 2018:Xie, 2018:Khan_addopt, 2018:Khan_AB, 2021:Nedic_PushPull, 2020:Doostmohammadian_Charalambous}, averaged over $20$ randomly generated strongly connected digraphs of $20$ nodes each. 
}
\label{comparisons_literature_optim}
\end{figure}

\section{Conclusions and Future Directions}\label{sec:conclusions}

In this paper, we considered the problem of distributed optimization for large-scale networks for the case where each node is endowed with a quadratic local cost function. 
We proposed a fast distributed algorithm, which operates over large-scale networks and converges in a finite number of time steps. 
We showed that our algorithm converges in a finite number of time steps to the  \textit{proximity of the} optimal solution, with the distance from the exact solution depending on the quantization level, and exhibits distributed stopping capability. 
The operation of our algorithm relies on event-triggered updates and each node processes and transmits quantized values. 
Furthermore, we presented a fully asynchronous algorithm which operates by performing max-consensus in an asynchronous fashion. 
We presented applications of our proposed algorithms for task scheduling over data centers, and global model aggregation over federated learning systems. 
In these applications we demonstrated the performance of our proposed algorithms and we have shown the algorithms' fast convergence by using extensive empirical evaluations. 
Finally, we showed that our algorithms compare favorably to algorithms in the literature. 

In the future we plan to extend our algorithms' operation for the case where each node wishes to maintain its privacy and not reveal the initial state it contributes to the computation of the optimal solution. 



\appendices
%
%
%
%
\section{Proof of Theorem~\ref{thm:main}}
\label{appendix:A}

We first consider Lemma~\ref{Lemma1_prob}, \emph{mutatis mutandis}, which is necessary for our subsequent development. 
Then, we consider Theorem~\ref{Theorem2_Alg2_Converge} (due to space considerations we provide a sketch of the proof which is an adaptation of the proof of Theorem~$1$ in \cite{2021:Rikos_Hadj_Johan}). 
Then, we present the proof of Theorem~\ref{thm:main}.


\begin{lemma}[\hspace{-0.00001cm}\cite{2020:RikosHadj_IFAC}]
\label{Lemma1_prob}
Consider a strongly connected digraph $\mathcal{G}_d = (\mathcal{V}, \mathcal{E})$ with $n=|\mathcal{V}|$ nodes and $m=|\mathcal{E}|$ edges.
Suppose that each node $v_j$ assigns a nonzero probability $b_{lj}$ to each of its outgoing edges $m_{lj}$, where $v_l \in \mathcal{N}^+_j \cup \{v_j\}$, as follows  
\begin{align*}
b_{lj} = \left\{ \begin{array}{ll}
         \frac{1}{1 + \mathcal{D}_j^+}, & \mbox{if $l = j$ or $v_{l} \in \mathcal{N}_j^+$,}\\
         0, & \mbox{if $l \neq j$ and $v_{l} \notin \mathcal{N}_j^+$.}\end{array} \right. 
\end{align*}
At time step $k=0$, node $v_j$ holds a ``token" while the other nodes $v_l \in \mathcal{V} - \{ v_j \}$ do not. 
Each node $v_j$ transmits the ``token'' (if it has it, otherwise it performs no transmission) according to the nonzero probability $b_{lj}$ it assigned to its outgoing edges $m_{lj}$. 
The probability $P^{D}_{T_i}$ that the token is at node $v_i$ after $D$ time steps (where $D$ is the diameter of the digraph $\mathcal{G}_d$, for which it holds that $D \leq n-1$) satisfies 
$
P^{D}_{T_i} \geq (1+\mathcal{D}^+_{max})^{-D} > 0 , 
$ 
where $\mathcal{D}^+_{max} = \max_{v_j \in \mathcal{V}} \mathcal{D}^+_{j}$. 
\end{lemma}


\begin{theorem}\label{Theorem2_Alg2_Converge}
Consider a strongly connected digraph $\mathcal{G}_d = (\mathcal{V}, \mathcal{E})$ with $n=|\mathcal{V}|$ nodes and $m=|\mathcal{E}|$ edges. 
At time step $k=0$, each node $v_j$
knows $z_j[0]$, $y_j[0]$. 
Suppose that each node $v_j \in \mathcal{V}$ follows the Initialization and Iteration steps as described in Algorithm~\ref{algorithm1}. 
For any $\varepsilon$, where $0 < \varepsilon < 1$, there exists $k_0 \in \mathbb{N}$, so that with probability $(1-\varepsilon)^{(y^{init}+n)}$ we have
\begin{equation}\label{tasks_convergence}
( q^s_j[k] = \lfloor q^{\mathrm{tasks}} \rfloor \ , \ \ k \geq k_0) \ \ \mathrm{or} \ \ ( q^s_j[k] = \lceil q^{\mathrm{tasks}} \rceil \ , \ \ k \geq k_0)  ,
\end{equation}
for every $v_j \in \mathcal{V}$, where 
\begin{equation}\label{initial_dist_processors}
q^{\mathrm{tasks}} = \frac{\sum_{j=1}^{n} y_j[0]}{\sum_{j=1}^{n} z_j[0]} , 
\end{equation} 
and 
\begin{align}
y^{init} & = & \sum_{\{v_j \in \mathcal{V}: y_j[0] > \lceil q^{\mathrm{tasks}} \rceil\}} {(y_j[0] - \lceil q^{\mathrm{tasks}} \rceil) } \ + \nonumber \\ 
 & & \sum_{\{v_j \in \mathcal{V}: y_j[0] < \lfloor q^{\mathrm{tasks}} \rfloor\}} {(\lfloor q^{\mathrm{tasks}} \rfloor - y_j[0])} , \label{initial_distance_no_oscill} 
\end{align}
is the total initial state error.
\end{theorem}

\begin{proof}
The operation of Algorithm~\ref{algorithm1} can be interpreted as the ``random walk'' of 
$\sum_{j=1}^{n} z_j[0] - n$ 
``tokens'' in a Markov chain.  
Specifically, at time step $k=0$, node $v_j$ holds 
$z_j[0]$
``tokens". 
One token is $T_j^{ins}$ and is stationary, whereas the other 
$z_j[0] - 1$
tokens are $T_j^{out, \vartheta}$, where $\vartheta = 1, 2, ..., z_j[0] - 1$, and perform independent random walks. 
Each token $T_j^{ins}$ and $T_j^{out, \vartheta}$ contains a pair of values $y_j^{ins}[k]$, $z_j^{ins}[k]$, and $y_j^{out, \vartheta}[k]$, $z_j^{out, \vartheta}[k]$, where $\vartheta = 1, 2, ..., z_j[0] - 1$, respectively. 
Initially, we have (i) $y_j^{ins}[0] = \lceil y_j[0] / z_j[0] \rceil$, (ii) $y_j^{out, \vartheta}[0] = \lceil y_j[0] / z_j[0] \rceil$ or $y_j^{out, \vartheta}[0] = \lfloor y_j[0] / z_j[0] \rfloor$ and (iii) $z_j^{ins}[0] = z_j^{out, \vartheta}[0] = 1$ for $\vartheta = 1, 2, ..., z_j[0] - 1$, such that 
$
y_j^{ins}[0] + \sum_{\vartheta = 1}^{z_j[0] - 1} y_j^{out, \vartheta}[0] = y_j[0], 
$
and
$
z_j^{ins}[0] + \sum_{\vartheta = 1}^{z_j[0] - 1} z_j^{out, \vartheta}[0] = z_j[0] . 
$
At each time step $k$, each node $v_j$ keeps the token $T_j^{ins}$ (i.e., it never transmits it) while it transmits the tokens $T_j^{out, \vartheta}$, where $\vartheta = 1, 2, ..., z_j[0] - 1$, independently to out-neighbors according to the nonzero probability $b_{lj}$ it assigned to its outgoing edges $m_{lj}$ during the Initialization Steps. 
If $v_j$ receives one or more tokens $T_i^{out, \vartheta}$ from its in-neighbors $v_i$, the values $y_i^{out, \vartheta}[k]$ and $y_j^{ins}[k]$ become equal (or with maximum difference equal to $1$); then $v_j$ transmits each received token $T_i^{out, \vartheta}$ to a randomly selected out-neighbor according to the nonzero probability $b_{lj}$ it assigned to its outgoing edges $m_{lj}$. 
Note here that during the operation of Algorithm~\ref{algorithm1} we have 
\begin{equation}\label{sum_preserve}
\sum_{j=1}^n \sum_{\vartheta = 1}^{z_j[0] - 1} y^{out, \vartheta}_j[k] + \sum_{j=1}^n y^{ins}_j[k] = \sum_{j=1}^n y_j[0] , \ \forall k \in \mathbb{Z}_+ .
\end{equation}

The main idea of this proof is that one token $T^{out, \vartheta}_{\lambda}$ visits a specific node $v_i$ (for which it holds $| y^{out, \vartheta}_{\lambda} - y^{ins}_i | > 1$) and obtains equal values $y$ (or with maximum difference between them equal to $1$) with the token $T^{ins}_i$ which is kept at node $v_i$. 
Thus, we analyze the required time steps for the specific token $T^{out, \vartheta}_{\lambda}$ (which performs a random walk) to visit node $v_i$ according to a probability. 
Note here that if each token $T^{out, \vartheta}_{\lambda}$ visits $y^{init}$ times each node $v_i$, then every token in the network (including both the tokens performing random walk and the stationary tokens) obtains $y$ value equal to $\lfloor q^{\mathrm{tasks}} \rfloor$ or $\lceil q^{\mathrm{tasks}} \rceil$.  

From Lemma~\ref{Lemma1_prob} we have that the probability $P^{D}_{T^{out}}$ that ``the specific token $T_{\lambda}^{out, \vartheta}$ is at node $v_i$ after $D$ time steps'' is 
\begin{equation}\label{lowerProf_no_oscil}
P^{D}_{T^{out}} \geq (1+\mathcal{D}^+_{max})^{-D} .
\end{equation}
This means that the probability $P^{D}_{N\_ T^{out}}$ that ``the specific token $T_{\lambda}^{out, \vartheta}$ has not visited node $v_i$ after $D$ time steps'' is
\begin{equation}\label{lowerProf_not_no_oscil}
P^{D}_{N\_ T^{out}} \leq 1- (1+\mathcal{D}^+_{max})^{-D} .
\end{equation}
By extending this analysis, we can state that for any $\varepsilon$, where $0 < \varepsilon < 1$ and after $\tau D$ time steps where
\begin{equation}\label{windows_for_conv_1_no_oscil}
\tau \geq \Big \lceil \dfrac{\log{\varepsilon}}{\log{(1 - (1+\mathcal{D}^+_{max})^{-D})}} \Big \rceil ,
\end{equation}
the probability $P^{\tau}_{N\_ T^{out}}$ that ``the specific token $T_{\lambda}^{out, \vartheta}$ has not visited node $v_i$ after $\tau D$ time steps'' is
\begin{equation}\label{ProbNotMeet_after_t_no_oscil}
P^{\tau}_{N\_ T^{out}} \leq [P^{D}_{N\_ T^{out}}]^{\tau} \leq \varepsilon .
\end{equation}
This means that after $\tau D$ time steps, where $\tau$ fulfills \eqref{windows_for_conv_1_no_oscil}, the probability that ``the specific token $T_{\lambda}^{out, \vartheta}$ has visited node $v_i$ after $\tau D$ time steps'' is equal to $1-\varepsilon$.

Thus, by extending this analysis, for $k \geq (y^{init} + n) \tau D$, where $y^{init}$ fulfills \eqref{initial_distance_no_oscill} and $\tau$ fulfills \eqref{windows_for_conv_1_no_oscil}, we have 
$q^s_j[k] = \lfloor q^{\mathrm{tasks}} \rfloor$ or $q^s_j[k] = \lceil q^{\mathrm{tasks}} \rceil$
with probability $(1-\varepsilon)^{(y^{init} + n)}$, for every $v_j \in \mathcal{V}$. 
\end{proof}


\noindent
\textit{Proof of Theorem~\ref{thm:main}.}
From Theorem~\ref{Theorem2_Alg2_Converge} we have that the operation of Algorithm~\ref{algorithm1} can be interpreted as the ``random walk'' of 
$\sum_{j=1}^{n} z_j[0] - n$ 
``tokens'' in a Markov chain. 
Furthermore, we also have that $n$ ``tokens'' remain stationary, one token at each node. 
Each of these 
$\sum_{j=1}^{n} z_j[0] - n$ 
tokens contains a pair of values $y^{out, \vartheta}[k]$, $z^{out, \vartheta}[k]$, where $\vartheta = 1, 2, ..., \sum_{j=1}^{n} z_j[0] - n$ and each of the $n$ stationary tokens contains a pair of values $y^{ins}[k]$, $z^{ins}[k]$. 
From Theorem~\ref{Theorem2_Alg2_Converge} we have that after $(y^{init} + n) \tau D$ time steps, where $y^{init}$ fulfills \eqref{initial_distance_no_oscill} and $\tau$ fulfills \eqref{windows_for_conv_1_no_oscil}, the state $q^s_j[k]$ of each node $v_j$ becomes $q^s_j[k] = \lfloor q^{\mathrm{tasks}} \rfloor$ or $q^s_j[k] = \lceil q^{\mathrm{tasks}} \rceil$
with probability $(1-\varepsilon)^{(y^{init} + n)}$, where $0 < \varepsilon < 1$, and $q^{\mathrm{tasks}}$ fulfills \eqref{initial_dist_processors}. 
This means that, after $(y^{init} + n) \tau D$ time steps, where $y^{init}$ fulfills \eqref{initial_distance_no_oscill} and $\tau$ fulfills \eqref{windows_for_conv_1_no_oscil}, for each of the 
$\sum_{j=1}^{n} z_j[0] - n$ 
tokens in the network it holds that $y^{out, \vartheta}[k] = \lfloor q^{\mathrm{tasks}} \rfloor$ or $y^{out, \vartheta}[k] = \lceil q^{\mathrm{tasks}} \rceil$, $\vartheta = 1, 2, ..., \sum_{j=1}^{n} z_j[0] - n$
while for each of the $n$ stationary tokens in the network it also holds that $y^{ins}[k] = \lfloor q^{\mathrm{tasks}} \rfloor$ or $y^{ins}[k] = \lceil q^{\mathrm{tasks}} \rceil$, with probability $(1-\varepsilon)^{(y^{init} + n)}$, where $0 < \varepsilon < 1$. 
Specifically, the $y$ value of every token in the network is equal either to $\lfloor q^{\mathrm{tasks}} \rfloor$ or $\lceil q^{\mathrm{tasks}} \rceil$ after $(y^{init} + n) \tau D$ time steps with probability $(1-\varepsilon)^{(y^{init} + n)}$, where $0 < \varepsilon < 1$. 

During the operation of Algorithm~\ref{algorithm1}, every $D$ time steps each node $v_j$ re-initializes its voting variables $M_j$, $m_j$ to be $M_j = \lceil y_j[k] / z_j[k] \rceil$, $m_j = \lfloor y_j[k] / z_j[k] \rfloor$. 
Note here that the $\max$-consensus algorithm (or the $\min$-consensus algorithm) converges to the maximum value among all nodes in a finite number of steps $s$, where $s \leq D$ (see, e.g., \cite[Theorem 5.4]{2013:Giannini}).
Thus, after $(y^{init} + n) \tau D$ time steps the value $q^s_j[k]$ of each node $v_j$ is equal to $\lfloor q^{\mathrm{tasks}} \rfloor$ or $\lceil q^{\mathrm{tasks}} \rceil$, with probability $(1-\varepsilon)^{(y^{init} + n)}$, where $0 < \varepsilon < 1$. 
This means that $M_j$, $m_j$ are re-initialized to be equal to $M_j = \lfloor q^{\mathrm{tasks}} \rfloor$ or $M_j = \lceil q^{\mathrm{tasks}} \rceil$ and $m_j = \lfloor q^{\mathrm{tasks}} \rfloor$ or $m_j = \lceil q^{\mathrm{tasks}} \rceil$ after $\lceil ((y^{init} + n) \tau D / D) \rceil D$ time steps with probability $(1-\varepsilon)^{(y^{init} + n)}$, where $0 < \varepsilon < 1$. 
After an additional number of $D$ time steps the variables $M_j$, $m_j$ of each node are updated to $M_j = \lfloor q^{\mathrm{tasks}} \rfloor$ and $m_j = \lfloor q^{\mathrm{tasks}} \rfloor$ (since, the $\max-$consensus algorithm \cite{2008:Cortes} converges after $D$ time steps). 
Thus, $M_j - m_j \leq 1$ holds for every node $v_j$. 
This means that every node $v_j$ calculates the optimal 
$x^*$ (shown in \eqref{eq:closedform}) 
and terminates its operation. 
As a result, we have that after $\lceil ((y^{init} + n) \tau D / D) \rceil D + D$ time steps each node $v_j$ calculates the optimal 
$x_j^* = x^* = \lceil y_j[0]/ q^{\mathrm{tasks}} \rceil$
with probability $(1-\varepsilon)^{(y^{init} + n)}$, where $0 < \varepsilon < 1$. \hspace*{\fill} $\square$

%
%
%
%
\section{Proof of Theorem~\ref{thm:main_2}}
\label{appendix:B}

We first consider Lemma~\ref{Lemma1_prob_asynch}, and Theorem~\ref{Theorem2_Alg2_Converge_asynch} which are necessary for our subsequent development. 
Then, we present the proof of Theorem~\ref{thm:main_2}. 

\begin{lemma}
\label{Lemma1_prob_asynch}
Consider a strongly connected digraph $\mathcal{G}_d = (\mathcal{V}, \mathcal{E})$ with $n=|\mathcal{V}|$ nodes and $m=|\mathcal{E}|$ edges.
Suppose that each node $v_j$ assigns a nonzero probability $b_{lj}$ to each of its outgoing edges $m_{lj}$, where $v_l \in \mathcal{N}^+_j \cup \{v_j\}$, as follows  
\begin{align*}
b_{lj} = \left\{ \begin{array}{ll}
         \frac{1}{1 + \mathcal{D}_j^+}, & \mbox{if $l = j$ or $v_{l} \in \mathcal{N}_j^+$,}\\
         0, & \mbox{if $l \neq j$ and $v_{l} \notin \mathcal{N}_j^+$.}\end{array} \right. 
\end{align*}
At time step $k=0$, node $v_j$ holds a ``token" while the other nodes $v_l \in \mathcal{V} - \{ v_j \}$ do not. 
Each node $v_j$ transmits the ``token'' (if it has it, otherwise it performs no transmission) according to the nonzero probability $b_{lj}$ it assigned to its outgoing edges $m_{lj}$. 
Furthermore, each node $v_j$ requires at most $\mathcal{B}$ time steps to process the information in the received token from its in-neighbors. 
The integer number of time steps that each node $v_j$ requires to process the information is a bounded discrete random variable with some distribution. 
Specifically, node $v_j$ requires $\lambda$ time steps, where $\lambda \in \{ 1, 2, ..., \mathcal{B} \}$, for processing information with probability $\mathcal{B}^{(\lambda)}_{j}$, where $\sum_{\lambda = 1}^{\mathcal{B}} \mathcal{B}^{(\lambda)}_{j} = 1$, for every $v_j$. 
The probability $P^{\mathcal{B} D}_{T_i}$ that the token is at node $v_i$ after $\mathcal{B} D$ (note that $D$ is the diameter of the digraph $\mathcal{G}_d$ and it holds that $D \leq n-1$) time steps satisfies 
$
P^{\mathcal{B} D}_{T_i} \geq (1+\mathcal{D}^+_{max})^{-D} (\mathcal{B}^{(\mathcal{B})}_{min})^{D}  > 0 , 
$
where $\mathcal{D}^+_{max} = \max_{v_j \in \mathcal{V}} \mathcal{D}^+_{j}$, and $\mathcal{B}^{(\mathcal{B})}_{min} = \min_{v_j \in \mathcal{V}} \mathcal{B}^{(\mathcal{B})}_{j}$.
\end{lemma}

\begin{proof}
We have that the diameter $D$ of every strongly connected digraph $\mathcal{G}_d$ is upper bounded by $n-1$, where $n=|\mathcal{V}|$. 
Specifically, from node $v_j$ to node $v_i$, there exists a sequence of nodes $v_j \equiv v_{l_0},v_{l_1}, \dots, v_{l_t} \equiv v_i$, such that $(v_{l_{\tau+1}},v_{l_{\tau}}) \in \mathcal{E}$ for $ \tau = 0, 1, \dots , t-1$, where $t \leq D$. 
This means that the shortest path from node $v_j$ to node $v_i$ ($v_j \neq v_i$) has length at most $D$. 
Let us assume that the token at time step $k=0$ is at node $v_j$. 
In one scenario, node $v_j$ processes the information of the token for $\mathcal{B}$ time steps with probability $(\mathcal{B}^{(\mathcal{B})}_{min})$. 
Then, at time step $\mathcal{B}$ node $v_j$ selects node $v_{l_1}$, with probability at least $(1+\mathcal{D}^+_{max})^{-1}$. 
This means that the token will be at node $v_{l_1}$ after $\mathcal{B}$ time steps with probability at least $(1+\mathcal{D}^+_{max})^{-1} (\mathcal{B}^{(\mathcal{B})}_{min})$. 
Again, in the worst case scenario, node $v_{l_1}$ processes the information of the received token for $\mathcal{B}$ time steps with probability $(\mathcal{B}^{(\mathcal{B})}_{min})$. 
Then, it transmits the token to node $v_{l_2}$ with probability at least $(1+\mathcal{D}^+_{max})^{-1}$. 
This means that the token will be at node $v_{l_2}$ after $2\mathcal{B}$ time steps with probability at least $(1+\mathcal{D}^+_{max})^{-2} (\mathcal{B}^{(\mathcal{B})}_{min})^{2}$. 
By repeating this analysis, we have that after $\mathcal{B}(t-1)$ time steps, the token will be at node $v_i$ with probability at least 
$
P^{\mathcal{B}(t-1)}_{T_i} \geq (1+\mathcal{D}^+_{max})^{-(t-1)} (\mathcal{B}^{(\mathcal{B})}_{min})^{(t-1)} > 0 . 
$
For the remaining $\mathcal{B}(D-t)$ time steps, we have that node $v_i$ processes the information for $\mathcal{B}$ time steps and then transmits it to itself. 
So, as a result, we have that after $\mathcal{B} D$ time steps, the token will be at node $v_i$ with probability at least 
$
P^{\mathcal{B} D}_{T_i} \geq (1+\mathcal{D}^+_{max})^{- D} (\mathcal{B}^{(\mathcal{B})}_{min})^{D} > 0 . 
$
This completes the proof of our lemma. 
\end{proof}



\begin{theorem}\label{Theorem2_Alg2_Converge_asynch}
Consider a strongly connected digraph $\mathcal{G}_d = (\mathcal{V}, \mathcal{E})$ with $n=|\mathcal{V}|$ nodes and $m=|\mathcal{E}|$ edges. 
At time step $k=0$, each node $v_j$ 
knows $z_j[0]$, $y_j[0]$. 
Furthermore, each node $v_j$ requires at most $\mathcal{B}$ time steps to process the information in the received token from its in-neighbors. 
The number of time steps that each node $v_j$ requires to process the information follows a random probability distribution. 
Specifically, node $v_j$ requires $\lambda$ time steps, where $\lambda \in \{ 1, 2, ..., \mathcal{B} \}$ for processing with probability $\mathcal{B}^{(\lambda)}_{j}$, where $\sum_{\lambda = 1}^{\mathcal{B}} \mathcal{B}^{(\lambda)}_{j} = 1$, for every $v_j$. 
Suppose that each node $v_j \in \mathcal{V}$ follows the Initialization and Iteration steps as described in Algorithm~\ref{algorithm2}. 
For any $\varepsilon$, where $0 < \varepsilon < 1$, there exists $k_0 \in \mathbb{N}$, so that with probability $(1-\varepsilon)^{(y^{init}+n)}$ we have
\begin{equation}\label{tasks_convergence_asynch}
( q^s_j[k] = \lfloor q^{\mathrm{tasks}} \rfloor \ , \ \ k \geq k_0) \ \ \mathrm{or} \ \ ( q^s_j[k] = \lceil q^{\mathrm{tasks}} \rceil \ , \ \ k \geq k_0)  ,
\end{equation}
for every $v_j \in \mathcal{V}$, where $q^{\mathrm{tasks}}$, $y^{init}$ fulfill \eqref{initial_dist_processors},  \eqref{initial_distance_no_oscill}. 
\end{theorem}

\begin{proof}
The proof is similar to the proof of Theorem~\ref{Theorem2_Alg2_Converge}. 
For this reason, we only mention the differences in comparison to the proof of Theorem~\ref{Theorem2_Alg2_Converge}.

From Lemma~\ref{Lemma1_prob_asynch} we have that the probability $P^{\mathcal{B} D}_{T^{out}}$ that ``the specific token $T_{\lambda}^{out, \vartheta}$ is at node $v_i$ after $\mathcal{B} D$ time steps'' is 
\begin{equation}\label{lowerProf_no_oscil_asynch}
P^{\mathcal{B} D}_{T^{out}} \geq (1+\mathcal{D}^+_{max})^{- D} (\mathcal{B}^{(\mathcal{B})}_{min})^{D} ,
\end{equation}
where $\mathcal{D}^+_{max} = \max_{v_j \in \mathcal{V}} \mathcal{D}^+_{j}$, $\mathcal{B}^{(\mathcal{B})}_{min} = \min_{v_j \in \mathcal{V}} \mathcal{B}^{(\mathcal{B})}_{j}$.
This means that the probability $P^{\mathcal{B} D}_{N\_ T^{out}}$ that ``the specific token $T_{\lambda}^{out, \vartheta}$ has not visited node $v_i$ after $\mathcal{B} D$ time steps'' is
\begin{equation}\label{lowerProf_not_no_oscil_asynch}
P^{\mathcal{B} D}_{N\_ T^{out}} \leq 1- (1+\mathcal{D}^+_{max})^{- D} (\mathcal{B}^{(\mathcal{B})}_{min})^{D} .
\end{equation}
By extending this analysis, we can state that for any $\varepsilon$, where $0 < \varepsilon < 1$ and after $\tau(\mathcal{B} D)$ time steps where
\begin{equation}\label{windows_for_conv_1_no_oscil_asynch}
\tau \geq \Big \lceil \dfrac{\log{\varepsilon}}{\log{(1 - (1+\mathcal{D}^+_{max})^{- D} (\mathcal{B}^{(\mathcal{B})}_{min})^{D})}} \Big \rceil ,
\end{equation}
the probability $P^{\tau}_{N\_ T^{out}}$ that ``the specific token $T_{\lambda}^{out, \vartheta}$ has not visited node $v_i$ after $\tau (\mathcal{B} D)$ time steps'' is $P^{\tau}_{N\_ T^{out}} \leq [P^{\mathcal{B} D}_{N\_ T^{out}}]^{\tau} \leq \varepsilon$. 
This means that after $\tau (\mathcal{B} D)$ time steps, where $\tau$ fulfills \eqref{windows_for_conv_1_no_oscil_asynch}, the probability that ``the specific token $T_{\lambda}^{out, \vartheta}$ has visited node $v_i$ after $\tau (\mathcal{B} D)$ time steps'' is $1-\varepsilon$.

Thus, by extending this analysis, for $k \geq (y^{init} + n) \tau (\mathcal{B} D)$, where $y^{init}$ fulfills \eqref{initial_distance_no_oscill} and $\tau$ fulfills \eqref{windows_for_conv_1_no_oscil_asynch}, we have $q^s_j[k] = \lfloor q^{\mathrm{tasks}} \rfloor$ or $q^s_j[k] = \lceil q^{\mathrm{tasks}} \rceil$ 
with probability $(1-\varepsilon)^{(y^{init} + n)}$, for every $v_j \in \mathcal{V}$. 
\end{proof}


\noindent
\textit{Proof of Theorem~\ref{thm:main_2}.}
The proof is similar to the proof of Theorem~\ref{thm:main}. 
For this reason, we only mention the differences in comparison to the proof of Theorem~\ref{thm:main}.

During the operation of Algorithm~\ref{algorithm2}, every $D \mathcal{B}$ time steps each node $v_j$ re-initializes its voting variables $M_j$, $m_j$ to be $M_j = \lceil y_j[k] / z_j[k] \rceil$, $m_j = \lfloor y_j[k] / z_j[k] \rfloor$. 
After $(y^{init} + n) \tau (\mathcal{B} D)$ time steps, where $y^{init}$ fulfills \eqref{initial_distance_no_oscill} and $\tau$ fulfills \eqref{windows_for_conv_1_no_oscil_asynch}, the value $q^s_j[k]$ of each node $v_j$ is equal to $\lfloor q^{\mathrm{tasks}} \rfloor$ or $\lceil q^{\mathrm{tasks}} \rceil$, with probability $(1-\varepsilon)^{(y^{init} + n)}$, where $0 < \varepsilon < 1$. 
This means that $M_j$, $m_j$ are re-initialized to be equal to $M_j = \lfloor q^{\mathrm{tasks}} \rfloor$ or $M_j = \lceil q^{\mathrm{tasks}} \rceil$ and $m_j = \lfloor q^{\mathrm{tasks}} \rfloor$ or $m_j = \lceil q^{\mathrm{tasks}} \rceil$ after $\lceil ((y^{init} + n) \tau (\mathcal{B} D) / D \mathcal{B}) \rceil D \mathcal{B}$ time steps with probability $(1-\varepsilon)^{(y^{init} + n)}$, where $0 < \varepsilon < 1$. 
After an additional number of $D \mathcal{B}$ time steps the variables $M_j$, $m_j$ of each node are updated to $M_j = \lfloor q^{\mathrm{tasks}} \rfloor$ and $m_j = \lfloor q^{\mathrm{tasks}} \rfloor$ (since, the asynchronous $\max-$consensus algorithm converges after $\mathcal{B} D$ time steps). 
Thus, the condition $M_j - m_j \leq 1$ holds for every node $v_j$. 
This means that every node $v_j$ calculates the optimal 
$x^*$ (shown in \eqref{eq:closedform})
and terminates its operation. 
As a result, we have that after $\lceil ((y^{init} + n) \tau (\mathcal{B} D) / D \mathcal{B}) \rceil D \mathcal{B} + D \mathcal{B}$ time steps each node $v_j$ calculates the optimal required task workload 
$x_j^* = x^* = \lceil y_j[0]/ q^{\mathrm{tasks}} \rceil$ 
with probability $(1-\varepsilon)^{(y^{init} + n)}$, where $0 < \varepsilon < 1$.  \hspace*{\fill} $\square$


\bibliographystyle{IEEEtran}
\bibliography{bibliography}

\end{document}